\documentclass[twocolumn,superscriptaddress,aps,pra]{revtex4-2}

\usepackage{graphicx}
\usepackage{amsmath}
\usepackage{amsthm}
\usepackage{amssymb}
\usepackage{latexsym}
\usepackage{array}
\usepackage{hyperref}
\usepackage{amsfonts}
\usepackage{dsfont}
\usepackage{mathrsfs}
\usepackage{verbatim}
\usepackage{bbold}
\usepackage[normalem]{ulem}
\usepackage{upgreek}
\usepackage{makecell}
\usepackage{adjustbox}
\usepackage{algorithm}
\usepackage{algpseudocode}
\usepackage{color}
\usepackage{bm}
\usepackage{times}
\usepackage{booktabs}
\usepackage{multirow}

\newcommand{\bra}[1]{\left<#1\right|}
\newcommand{\ket}[1]{\left|#1\right>}
\newcommand{\abs}[1]{\left|#1\right|}

\newcommand{\braket}[2]{\left<{#1}|{#2}\right>}
\newcommand{\ketbra}[2]{\ket{#1}\!\!\bra{#2}}
\newcommand{\tr}{\operatorname{Tr}}
\newcommand{\diag}{\operatorname{Diag}}

\newcommand{\Hcsm}{\mathsf{HCSM}}
\newcommand{\Ptms}{\mathsf{PTM}}
\newcommand{\Qtms}{\mathsf{QTM}}
\newcommand{\Senm}{\mathsf{SEnM}}

\newcommand{\id}{\mathrm{id}}

\newtheorem{theorem}{Theorem}
\newtheorem{proposition}{Proposition}
\newtheorem{lemma}{Lemma}
\newtheorem{corollary}{Corollary}
\newtheorem{definition}{Definition}

\newtheorem{remark}{Remark}

\begin{document}

\title{The Classical Boundary of Quantum Computation: History-Resolved Causal Stochastic Machines and the Cost of Classicalization}

\author{Kyoungho~Cho}
\affiliation{Institute for Convergence Research and Education in Advanced Technology, Yonsei University, Seoul 03722, Republic of Korea}
\affiliation{Department of Statistics and Data Science, Yonsei University, Seoul 03722, Republic of Korea}

\author{Junghee~Ryu}
\affiliation{Center for Quantum Information R\&D, Korea Institute of Science and Technology Information (KISTI), Daejeon 34141, Republic of Korea}
\affiliation{Division of Quantum Information, KISTI School, Korea University of Science and Technology, Daejeon 34141, Republic of Korea}

\author{Minsu~Kim}
\affiliation{Institute for Convergence Research and Education in Advanced Technology, Yonsei University, Seoul 03722, Republic of Korea}

\author{Jeongho~Bang}\email{jbang@yonsei.ac.kr}
\affiliation{Institute for Convergence Research and Education in Advanced Technology, Yonsei University, Seoul 03722, Republic of Korea}
\affiliation{Department of Quantum Information, Yonsei University, Incheon 21983, Republic of Korea}

\date{\today}

\begin{abstract}
What separates quantum computation from classical computation when a classical model may retain every branch history and arbitrary causal memory? We study this question using history-resolved causal stochastic machines (HCSMs), which describe positive causal history laws. Probabilistic Turing machines (PTMs) embed in stochastic-ensemble models (SEnMs), which in turn embed in HCSMs. We prove that HCSM history laws coincide exactly with the fixed-basis, prefix-preserving diagonal sector of ideal quantum processes, denoted by $\Qtms_{\rm proc}$ and distinguished here from algorithmic quantum Turing machines (QTMs). Coherent quantum processes extend strictly beyond this sector. The boundary is already visible in reversibility. Born-probability shadows of unitary gates preserve inverse composition exactly for basis permutations with phases. Outside this class, an exact positive stochastic simulation that supports universal quantum continuations must retain predictive information absent from visible probabilities. An inverse-recovery identity quantifies this phase ledger, and universal continuation forces the retained probability law to be tomographically sufficient for the quantum state. Under continuous full-state encoding assumptions, this requirement yields quantum-state dimension lower bounds. The familiar slogan \(\mathsf{PTM}\subseteq\mathsf{QTM}\) is thereby sharpened into a precise classical boundary and a no-free-classicalization principle. Complete classical branching does not remove the predictive information needed to represent quantum coherence.
\end{abstract}

\maketitle

\section{Introduction}\label{sec:intro}

The standard narrative says that quantum computation extends classical computation by allowing quantum superposition, phase, and interference~\cite{Feynman1986,Grover1997,Shor1997}. As a slogan this is correct, but as a mathematical statement it leaves open a more delicate structural question. If one strengthens a classical probabilistic model as far as possible while retaining ordinary positive probabilities, causal updating, and stochastic branch aggregation, does the resulting model approach quantum computation, or does it remain confined to a sharply identifiable classical sector?

This question is intentionally different from the usual complexity-theoretic problem of separating BPP and BQP~\cite{Yao1993,BernsteinVazirani1993}. We do not attempt to prove that efficient probabilistic computation is strictly weaker than efficient quantum computation. A sufficiently slow classical machine can, of course, multiply matrices and numerically reproduce the final output statistics of finite quantum circuits. That is not the issue studied here. We ask a structural question instead. Can a positive stochastic process model represent the same elementary pieces as quantum computation and preserve the way those pieces compose?

Earlier stochastic-ensemble analyses, including the SEnM framework of Ref.~\cite{Bang2018}, already suggested that final output agreement is an insufficient criterion, and that temporal or information-theoretic consistency conditions reveal a more meaningful classical/quantum divide. However, the SEnM analysis leaves open what information a positive stochastic representation must retain to preserve coherent quantum composition. Here, we address this question through an explicit history-based formulation and quantify the predictive information required for inverse recovery and future quantum continuations.

To that end, we define a ``history-resolved causal stochastic machine (HCSM).'' It is a strong classical model in which every branch history is retained explicitly and the next transition may depend on the entire past. HCSM therefore permits history dependence beyond a fixed Markovian or ensemble-snapshot description. This does not assert a computational separation from machines whose internal configurations are enlarged to store that history. Nor is HCSM a nondeterministic Turing machine (NTM). It cannot promote an accepting branch to certainty merely because one exists, and its branch aggregation remains affine in ordinary probability mass. This distinction isolates positive, affine stochastic branch aggregation. Here ``non-transcendental'' refers to the branch semantics, not to an independent computability assumption.

Several established frameworks inform this comparison. The probability chain rule supplies the history-law factorization~\cite{CoverThomas}. Classical channels also have an algebraic description as maps between commutative algebras~\cite{Coecke2016Channels}. Quantum combs and process tensors formalize causal composition with memory~\cite{Chiribella2009Networks,Pollock2018Processes}. Generalized extension theorems explain why interventions require consistency conditions beyond those of passive observations~\cite{Milz2020Extension}. Resource theories quantify basis-dependent coherence and distinguish operational classes of incoherent maps~\cite{Baumgratz2014Coherence,ChitambarGour2016Comparison}, while unistochastic-matrix theory describes probabilities obtained from unitary matrix elements~\cite{Bengtsson2005Unistochastic}. We bring these ingredients together to compare history retention, composition, and recovery within one specified classical interface. The fixed-basis diagonal history sector used here is narrower than the classes of quantum operations that merely preserve incoherent input states.

The simulation requirement also distinguishes our approach from related classicality criteria. Quasiprobability and contextuality approaches constrain positive representations under their operational assumptions~\cite{Spekkens2008Negativity,Veitch2012Negativity}, while dimensional bounds for ontological models depend on regularity and dynamical assumptions~\cite{Montina2011Dimension}. Cobucci \emph{et al.}~\cite{Cobucci2026Operational} instead characterize quantum state sets obtained by stochastic coordination of commuting preparation devices, with a basis that may vary between devices. In our setting, one retained representation must support causal composition, inverse recovery, and subsequent quantum continuations. The contribution is a joint characterization connecting the computational-model hierarchy to the information needed beyond visible probability closure. The history-law correspondence specifies the classical sector. Recovery and continuation then determine what an extension of that sector must retain. The dimension statements concern continuous encodings of the retained probability law, rather than the dimension of an individual sampled hidden state. Placing these questions in one operational setting turns model inclusion into a testable requirement on information preservation.

Our first set of results locates this stronger classical model within quantum process theory. Every finite-horizon HCSM process has a faithful realization as a diagonal density process, and its image is exactly the fixed-basis, prefix-preserving dephased history sector. Classical stochastic computation thus appears as a well-defined commutative, positive-probability subtheory of quantum processes. The next results explain why this inclusion marks a boundary. The probability shadow of a quantum gate looks stochastic under a single basis-to-basis readout, but generally fails to preserve the compositional and reversible structure of quantum dynamics. Repairing that failure requires additional information. We quantify it through a phase ledger and prove that, under universal exact compositional simulation, the probability law of the retained classical information must carry the full operational content of the quantum state.

This gives more content to the familiar slogan $\Ptms \subseteq \Qtms$. Within finite horizons and countable configuration spaces, HCSM already admits arbitrary classical branching, memory, and positive causal stochastic dependence. Adding unrecorded routes cannot repair the structural failure of the visible probability dynamics. Their predictive information has to be retained. That information lies in recoverable relations among branches, namely quantum phase or coherence. We therefore obtain a no-free-classicalization principle. Beyond visible probability closure, exact classicalization must retain the missing quantum predictive information or adopt different weight and composition rules.

\section{HCSM and the maximal classical stochastic sector}\label{sec:hcsm-sector}

\subsection{Process simulation rather than complexity separation}\label{subsec:process-sim} 

We begin by distinguishing output agreement from compositional process simulation. We reserve $\Qtms$ for algorithmic quantum Turing machines with effective finite machine descriptions and specified admissible amplitudes~\cite{Benioff1982,Deutsch1985}. In contrast, $\Qtms_{\rm proc}$ denotes ideal quantum processes, allowing states and CPTP maps on finite-dimensional or separable Hilbert spaces, with ancillas and discarding~\cite{Watrous2018}. The structural embeddings below take values in $\Qtms_{\rm proc}$. They do not assume that arbitrary real transition probabilities admit exact finite-gate implementations. The familiar algorithmic PTM--QTM inclusion is recovered separately for standard fair-coin machines, as stated in Corollary~\ref{cor:effective-PTM-QTM}. Appendix~\ref{app:diagonal-embedding} gives the broader implementation conditions.

\begin{definition}[Output simulation]
\label{def:output-simulation}
A classical model output-simulates a quantum circuit $C$ on input $x$ if it reproduces the final readout distribution,
\begin{eqnarray}
P_{\rm cl}(y | C, x) = P_{\rm Q}(y | C, x).
\end{eqnarray}
No requirement is imposed on intermediate structure or on how the simulator represents circuit composition.
\end{definition}

Output simulation alone is too weak for the present purpose because it does not constrain the representation of intermediate operations. Numerical amplitude tracking can also satisfy compositional simulation when the tracked amplitudes remain part of the retained classical configuration. The earlier SEnM hypothesis in Ref.~\cite{Bang2018} belongs to a stronger category, which we call temporal-transition simulation. Selected intermediate conditional readout probabilities and the final input-output probability must agree with those of the quantum computation model.

\begin{definition}[Compositional process simulation]
\label{def:compositional-process-simulation}
Let $\mathcal G$ be a family of quantum circuit fragments. On a fixed configuration space, a positive stochastic representation $F$, with specified preparation and readout rules, is a compositional process simulation on $\mathcal G$ if
\begin{eqnarray}
F(\hat{\mathds{1}})=\hat{\mathds{1}}, \quad F(\hat{V}\hat{U})=F(\hat{V})F(\hat{U})
\end{eqnarray}
whenever $\hat{U}, \hat{V}, \hat{V}\hat{U} \in \mathcal G$, and if the associated readout probabilities agree with the intended quantum probabilities in the chosen basis. When complete bookkeeping histories are retained, these identities are required on operational equivalence classes of reachable laws: two laws are equivalent when every allowed future continuation has the same readout distribution. An inverse may restore all predictions without erasing the gate history.
\end{definition}

This study works mainly at this level. A classical computer may still calculate amplitudes. The question is whether a positive stochastic process can represent the same elementary pieces and preserve their composition. At that level, interference cannot be hidden.

\subsection{History-resolved causal stochastic machines}\label{subsec:hcsm-def} 

We now define the central classical model.

\begin{definition}[History-resolved causal stochastic machine]
\label{def:hcsm}
Fix a finite time horizon $T$. For each time $t=0,\ldots,T$, let $X_t$ be a nonempty finite or countable configuration set. A history at time $t$ is
\begin{eqnarray}
h_t=(x_0,x_1,\ldots,x_t) \in \mathcal H_t:=X_0 \times \cdots \times X_t.
\end{eqnarray}
A history-resolved causal stochastic machine (HCSM) is a tuple
\begin{eqnarray}
{\mathcal M} = \bigl(\{X_t\}_{t=0}^T, \mu_0, \{K_{t+1}\}_{t=0}^{T-1}, \{R_t\}_{t=0}^T, A\bigr),
\end{eqnarray}
where $\mu_0 \in \Delta(X_0)$ is an initial distribution, $A \subseteq {\mathcal H}_T$ is an accepting set, $R_t:\mathcal H_t \to M_t$ is a readout map, and
\begin{eqnarray}
K_{t+1}(x_{t+1}|h_t) \geq 0,
\quad
\sum_{x_{t+1}\in X_{t+1}}K_{t+1}(x_{t+1}|h_t)=1
\end{eqnarray}
is a causal stochastic kernel. The induced history distribution is
\begin{eqnarray}
\mu_{t+1}(h_t,x_{t+1}) = \mu_t(h_t)K_{t+1}(x_{t+1}|h_t),
\label{eq:hcsm-update-v03}
\end{eqnarray}
and the acceptance probability is
\begin{eqnarray}
P_{\rm acc}(\mathcal M) = \sum_{h_T \in A}\mu_T(h_T).
\end{eqnarray}
\end{definition}

\begin{figure}[t]
\centering
\includegraphics[width=0.85\linewidth]{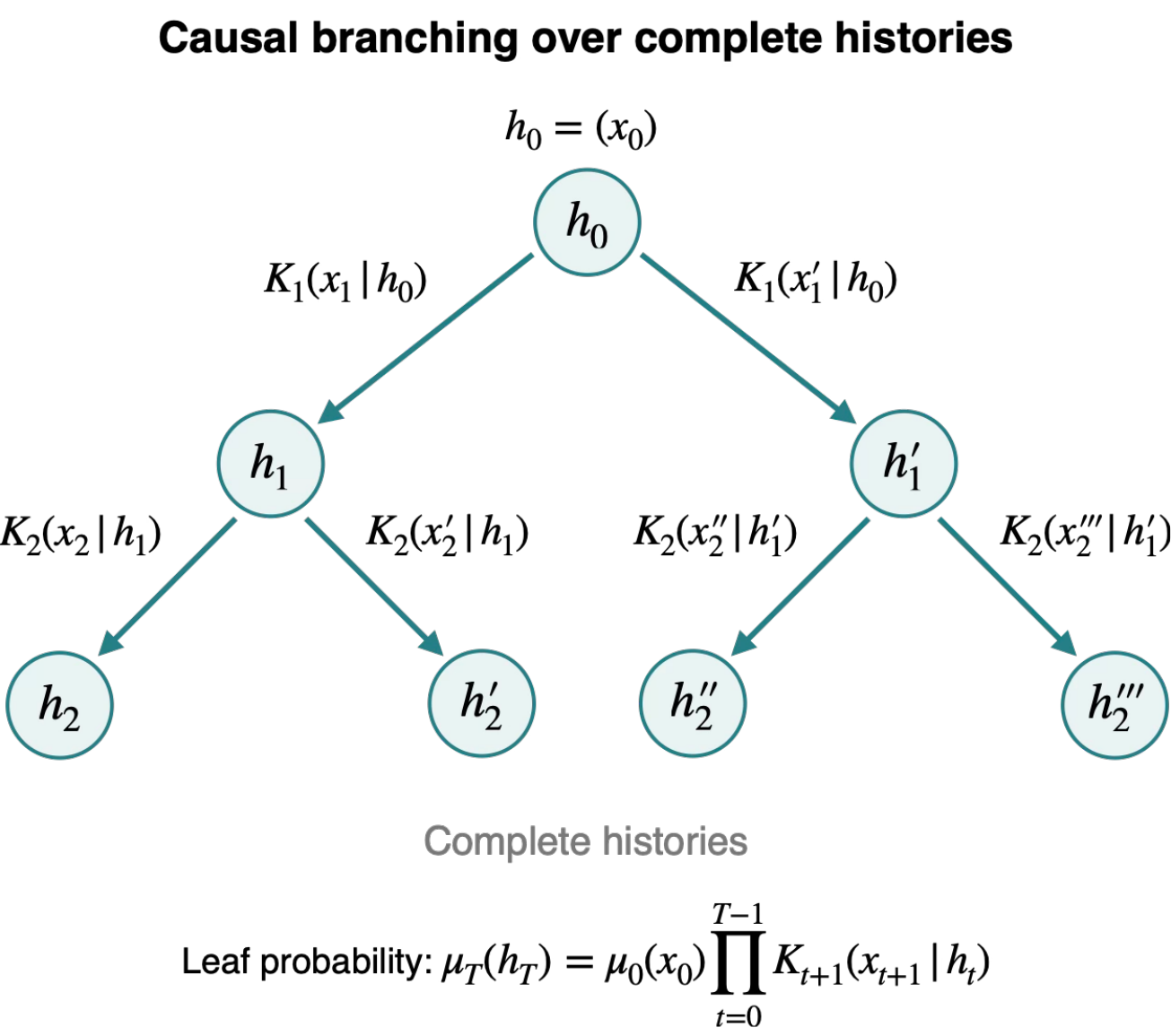}
\caption{An HCSM is a causal stochastic tree over full histories. Each edge is labeled by a kernel $K_{t+1}(x_{t+1}|h_t)$ that may depend on the entire previous history $h_t$. A leaf $h_T$ receives probability $\mu_T(h_T)=\mu_0(x_0)\prod_{t=0}^{T-1}K_{t+1}(x_{t+1} | h_t)$. This makes HCSM non-Markovian in general, while preserving ordinary positive probability.}
\label{fig:hcsm-tree}
\end{figure}

For externally selected controls, causality additionally requires nonanticipation: an update, initial law, or readout may use only controls already supplied. Kernels for control sequences sharing the same supplied prefix must agree up to that step. The chain-rule factorization of one separately fixed protocol does not by itself establish this cross-protocol condition.

The term ``history-resolved'' means that HCSM may keep the full branch history rather than only the current configuration. Causality restricts the next transition to dependence on the past. Stochasticity requires nonnegative probability weights and aggregation by summation of probability mass. Fig.~\ref{fig:hcsm-tree} depicts the history-tree form of the update rule in Eq.~\eqref{eq:hcsm-update-v03}.

\begin{theorem}[PTM embedding into HCSM]
\label{thm:PTMinHCSM}
Every finite-horizon probabilistic Turing-machine process is a special case of HCSM.
\end{theorem}

\begin{proof}---Let $x_t$ denote the PTM configuration at time $t$, including internal state, tape contents on the used region and head location. A PTM transition has the form $P_{\rm PTM}(x_{t+1}|x_t)$. Define an HCSM on the same configuration sets by
\begin{eqnarray}
K_{t+1}(x_{t+1}|h_t) := P_{\rm PTM}(x_{t+1}|x_t),
\end{eqnarray}
where $x_t$ is the last component of $h_t$. Substituting this into Eq.~\eqref{eq:hcsm-update-v03} gives exactly the PTM path distribution. Readout and acceptance sets are copied from the PTM. Hence, the PTM process is an HCSM process with Markovian kernels.
\end{proof}

\begin{theorem}[SEnM embedding into HCSM]
\label{thm:SEnMinHCSM}
Every finite-horizon SEnM process whose ensemble-configuration laws have finite or countable support is representable as an HCSM process. For the comparison in this paper, the ensemble class includes its singleton case.
\end{theorem}

\begin{proof}---In SEnM, let $s_j$ denote the ensemble state collection at step $j$, and let the non-Markovian ensemble transition be $C(s_j|s_{j-1},\ldots,s_0)$, as in the construction of the joint probability $P(s_0,\ldots,s_j)$ in Ref.~\cite{Bang2018}. Set the HCSM configuration space $X_j$ equal to the stated countable set of supported ensemble collections $s_j$, and define
\begin{eqnarray}
K_j(x_j|h_{j-1}) = C(s_j|s_{j-1},\ldots,s_0).
\end{eqnarray}
On ensemble histories of zero probability, choose any normalized kernel on $X_j$. The HCSM history distribution then reproduces the SEnM joint probability over ensemble collections. The SEnM readout probabilities are recovered by choosing readout maps or readout kernels corresponding to the SEnM ensemble readout rule. Thus, SEnM is included in HCSM. Under the singleton convention, take $s_j=(x_j)$, use the PTM transition $P_{\rm PTM}(x_j|x_{j-1})$, and retain its single-copy readout; the entire PTM path law is reproduced. This supplies $\Ptms\subseteq\Senm\subseteq\Hcsm$. The countability and singleton conventions specify the comparison made here; arbitrary uncountable ensemble spaces are not covered by the countable HCSM construction.
\end{proof}

\begin{theorem}[HCSM is not NTM semantics]
\label{thm:hcsm-not-ntm}
HCSM does not realize nondeterministic existential branch selection. In particular, its acceptance functional is affine in the history distribution, whereas NTM acceptance is not.
\end{theorem}

\begin{proof}---For a fixed accepting set $A$, HCSM acceptance is
\begin{eqnarray}
\mathcal A(\mu) = \sum_{h \in A}\mu(h).
\end{eqnarray}
Therefore, for probability distributions $\mu,\nu$ and $0 \leq \lambda \leq 1$,
\begin{eqnarray}
\mathcal A(\lambda\mu+(1-\lambda)\nu) = \lambda\mathcal A(\mu) + (1-\lambda)\mathcal A(\nu).
\end{eqnarray}
By contrast, NTM acceptance is represented at the support level by
\begin{eqnarray}
\mathcal A_{\exists}(\mu) = \mathbf 1_{\mu(A)>0}.
\end{eqnarray}
Take $\mu(A)=1$, $\nu(A)=0$ and $0<\lambda<1$. Then, $\mathcal A_{\exists}(\lambda\mu+(1-\lambda)\nu)=1$, while $\lambda{\mathcal A}_{\exists}(\mu) + (1-\lambda)\mathcal A_{\exists}(\nu)=\lambda$. Thus, the NTM functional is not affine. HCSM may retain all histories, but it cannot accept merely because some branch exists.
\end{proof}

HCSM makes the route history part of the retained description. One-time ensemble proportions do not in general determine temporal route correlations, whereas the HCSM path law does. This distinction is relative to the chosen snapshot description: an ensemble copy may itself be enlarged to store its history. Appendix~\ref{app:senm-vs-hcsm} state this comparison explicitly.

\subsection{Normal form and readout consistency}\label{subsec:maximality} 

The next result explains why HCSM is not just a convenient model. It is a normal form for finite-horizon positive causal history processes.

\begin{widetext}
\begin{theorem}[HCSM normal form]
\label{thm:hcsm-normal-form}\label{app:thm:normal-form-countable}
Let $P(x_0,\ldots,x_T)$ be any joint probability distribution on a finite or countable product space. Then, there exist causal kernels $K_{t+1}(x_{t+1}|x_0,\ldots,x_t)$ and an initial distribution $\mu_0$ such that
\begin{eqnarray}
P(x_0,\ldots,x_T) = \mu_0(x_0) \prod_{t=0}^{T-1}K_{t+1}(x_{t+1}|x_0,\ldots,x_t).
\end{eqnarray}
\end{theorem}
\end{widetext}

\begin{proof}---Set $\mu_0(x_0)=P(x_0)$. Whenever $P(x_0,\ldots,x_t)>0$, define
\begin{eqnarray}
K_{t+1}(x_{t+1}|x_0,\ldots,x_t) = \frac{P(x_0,\ldots,x_t,x_{t+1})}{P(x_0,\ldots,x_t)}.
\end{eqnarray}
On the zero-probability histories, choose any normalized kernel. The statement is simply the chain rule for probability, with nonnegativity and normalization inherited from conditional probability.
\end{proof}

HCSM is therefore a normal form for finite-horizon positive laws on countable classical histories. Within this scope, enlarging classical memory or allowing arbitrary dependence on the available past does not take a process outside HCSM. For controlled families, the nonanticipation condition must also hold. Computability and extensions to uncountable state spaces are separate questions.

\begin{theorem}[Global readout joint distribution]
\label{thm:global-readout-joint}
Every HCSM induces a global joint probability distribution over any finite sequence of readouts $M_0,\ldots,M_T$.
\end{theorem}

\begin{widetext}
\begin{proof}---Since each readout is a function $R_t:\mathcal H_t\to M_t$, define
\begin{eqnarray}
P(m_0,\ldots,m_T) = \sum_{h_T\in\mathcal H_T}\mu_T(h_T)\prod_{t=0}^{T}\mathbf 1\{R_t(h_t)=m_t\},
\end{eqnarray}
where $h_t$ is the length-$(t+1)$ prefix of $h_T$. This is the pushforward of $\mu_T$ under the readout maps and is therefore a normalized joint probability distribution.
\end{proof}
\end{widetext}

\begin{corollary}[Readout inequality for HCSM]
\label{cor:hcsm-readout-inequality}
For every HCSM, with discrete conditional Shannon entropies understood as nonnegative sums,
\begin{eqnarray}
H(M_T|M_0) \leq \sum_{t=1}^{T}H(M_t | M_{t-1}).
\label{eq:readout-ineq-v03}
\end{eqnarray}
\end{corollary}

\begin{proof}---Because the global joint distribution exists by Theorem~\ref{thm:global-readout-joint},
\begin{eqnarray}
H(M_T | M_0) &\leq& H(M_1,\ldots,M_T | M_0) \nonumber \\
	&=& \sum_{t=1}^{T}H(M_t | M_0,\ldots,M_{t-1}) \nonumber \\
	&\leq& \sum_{t=1}^{T}H(M_t | M_{t-1}),
\end{eqnarray}
where the last step uses monotonicity of conditional entropy under additional conditioning. Thus, HCSM inherits the same classical consistency mechanism that underlies the readout inequality in Ref.~\cite{Bang2018}, but now at the level of full history-resolved path measures. All terms refer to one readout protocol. Probabilities from different intermediate measurement protocols can be compared with this bound only when they are required to be marginals of the same joint law.
\end{proof}

A violation of Eq.~\eqref{eq:readout-ineq-v03} assembled from different intermediate measurement protocols therefore excludes an HCSM simulation only if the tested statistics must be marginals of one common joint law. It does not, by itself, rule out a causal classical simulator that represents the interventions explicitly and allows them to alter subsequent readouts. Indeed, the outcomes of any single actual sequential quantum measurement protocol have a joint law and satisfy the same inequality. To establish the structural boundary, we must therefore identify the quantum processes in the canonical HCSM image and the information required to reproduce coherent continuations beyond that image.

\subsection{From HCSM to the dephased quantum sector}\label{subsec:diag-sector} 

We now pass from the classical boundary model to quantum process theory.

\begin{theorem}[HCSM embedding into ideal quantum processes]
\label{thm:HCSMtoQTM}
Every finite-horizon HCSM admits a faithful embedding into $\Qtms_{\rm proc}$ on Hilbert spaces spanned by orthogonal history states. Under this embedding, the HCSM distribution $\mu_t$ becomes a diagonal density matrix
\begin{eqnarray}
\hat{\rho}_t = \sum_{h_t \in {\mathcal H}_t} \mu_t(h_t)\ketbra{h_t}{h_t},
\label{eq:diag-state-v03}
\end{eqnarray}
and each HCSM transition is implemented by a CPTP map.
\end{theorem}

\begin{proof}---Let ${\mathscr H}_t = \ell^2({\mathcal H}_t)$ with basis $\{\ket{h_t}:h_t \in {\mathcal H}_t\}$. For each kernel $K_{t+1}(x|h)$ define Kraus operators
\begin{eqnarray}
\hat{A}_{h,x} = \sqrt{K_{t+1}(x|h)}\ketbra{h,x}{h}.
\end{eqnarray}
Then,
\begin{eqnarray}
\sum_{h,x}\hat{A}_{h,x}^{\dagger}\hat{A}_{h,x}=\sum_h\Bigl(\sum_xK_{t+1}(x|h)\Bigr)\ketbra{h}{h}=\hat{\mathds{1}},
\end{eqnarray}
so
\begin{eqnarray}
{\mathcal E}_{t+1}(\hat{\rho})=\sum_{h,x}\hat{A}_{h,x}\hat{\rho} \hat{A}_{h,x}^{\dagger}
\end{eqnarray}
is trace-preserving and completely positive~\cite{Kraus1983,Choi1975}. For countably infinite histories, the Kraus completeness sum converges strongly and the channel sum converges in trace norm on trace-class states. Applying it to Eq.~\eqref{eq:diag-state-v03} gives
\begin{eqnarray}
{\mathcal E}_{t+1}(\hat{\rho}_t) = \sum_{h,x}\mu_t(h)K_{t+1}(x|h)\ketbra{h,x}{h,x} = \hat{\rho}_{t+1},
\end{eqnarray}
with $\hat{\rho}_{t+1}$ corresponding exactly to the HCSM distribution $\mu_{t+1}$. Induction over time proves faithfulness of all history and readout probabilities.
\end{proof}

\begin{corollary}[PTM inclusion in ideal quantum processes]
\label{cor:PTMinQTM}
Since every PTM process is an HCSM process and every HCSM process embeds into an ideal quantum process,
\begin{eqnarray}
\Ptms \subseteq \Hcsm \hookrightarrow \Qtms_{\rm proc}.
\end{eqnarray}
\end{corollary}

\begin{proof}---The statement is immediate from Theorem~\ref{thm:PTMinHCSM} and the HCSM embedding into ideal quantum processes in Theorem~\ref{thm:HCSMtoQTM}.
\end{proof}

\begin{figure}[t]
\centering
\includegraphics[width=0.98\linewidth]{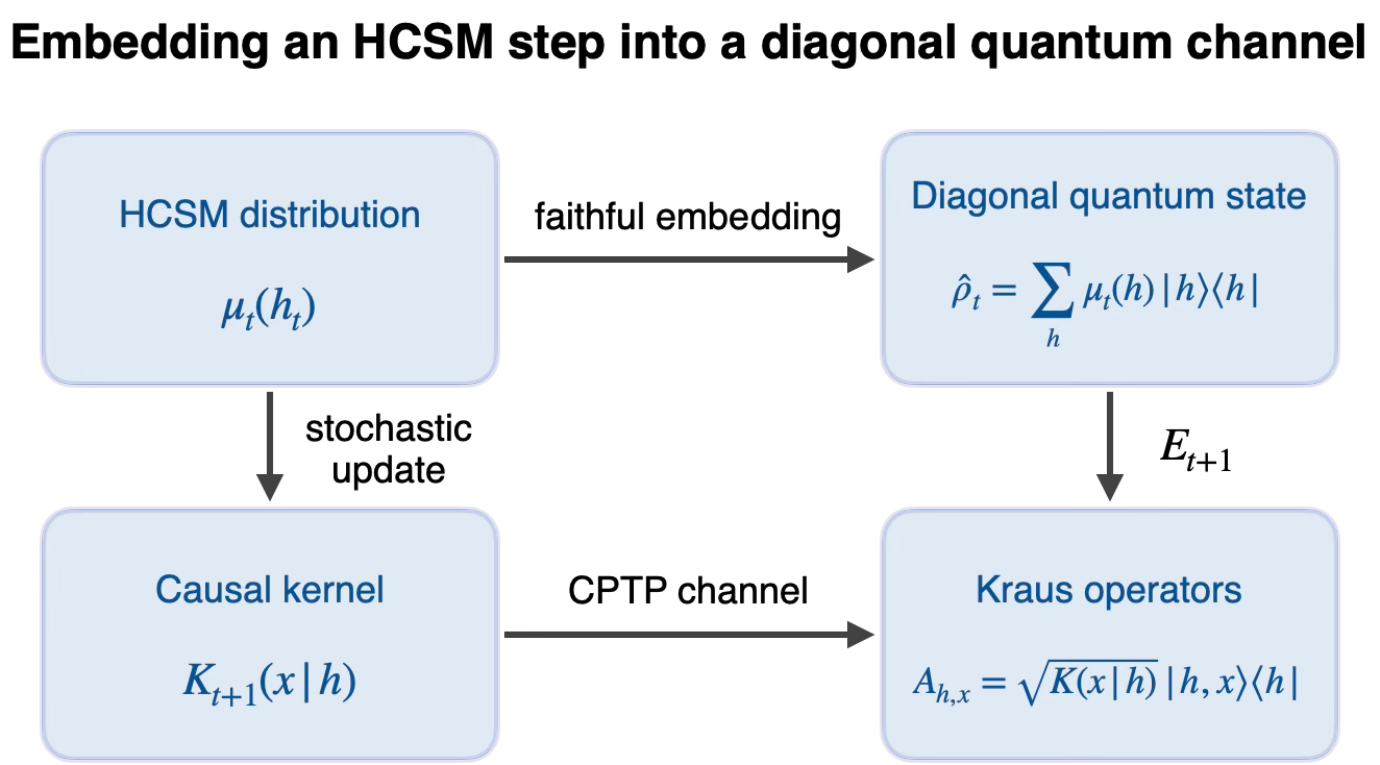}
\caption{Faithful diagonal embedding of one HCSM step. A probability distribution over histories is encoded as a diagonal density matrix over orthogonal history kets. A causal stochastic kernel becomes a quantum channel with Kraus operators $\hat{A}_{h,x}=\sqrt{K(x|h)}\ketbra{h,x}{h}$. Thus, every HCSM transition is a coherence-free quantum transition.}
\label{fig:embedding}
\end{figure}

Fig.~\ref{fig:embedding} illustrates the embedding. The proof can also be phrased unitarily: each stochastic kernel dilates to an isometry that stores the sampled successor in an orthogonal garbage register, and every such isometry extends to a unitary on a larger Hilbert space. The detailed dilation argument is given in Appendix~\ref{app:diagonal-embedding}.

\begin{definition}[Dephasing-fixed channel]
\label{def:dephasing-fixed-channel}
Fix computational bases for input and output Hilbert spaces and let $D_X$ and $D_Y$ denote the corresponding dephasing maps. A channel $\Phi:\mathcal T_1({\mathscr H}_X) \to {\mathcal T_1}({\mathscr H}_Y)$ on trace-class operators is dephasing-fixed if
\begin{eqnarray}
\Phi=D_Y \Phi D_X.
\end{eqnarray}
It neither uses input coherence nor creates output coherence in the selected bases.
\end{definition}

\begin{theorem}[Classical-sector characterization]
\label{thm:classical-sector}
For fixed finite or countable bases, dephasing-fixed quantum channels are exactly stochastic kernels written in quantum notation.
\end{theorem}

\begin{proof}---Suppose $\Phi=D_Y \Phi D_X$. Define
\begin{eqnarray}
K(y|x)=\bra{y}\Phi(\ketbra{x}{x})\ket{y}.
\end{eqnarray}
The complete positivity gives $K(y|x) \geq 0$ and the trace preservation gives $\sum_y K(y|x)=1$. For any diagonal state $\hat{\rho}_p=\sum_x p_x\ketbra{x}{x}$,
\begin{eqnarray}
\Phi(\hat{\rho}_p) = \sum_y \left(\sum_x K(y|x) p_x \right)\ketbra{y}{y},
\end{eqnarray}
so $\Phi$ acts as the stochastic transition $p \mapsto Kp$. Conversely, given a stochastic kernel $K$, define Kraus operators
\begin{eqnarray}
\hat{A}_{y,x} = \sqrt{K(y|x)} \ketbra{y}{x}.
\end{eqnarray}
They satisfy $\sum_{x,y}\hat{A}_{y,x}^{\dagger}\hat{A}_{y,x}=\hat{\mathds{1}}$ and generate a CPTP map $\Phi_K(\hat{\rho})=\sum_{x,y}\hat{A}_{y,x}\hat{\rho} \hat{A}_{y,x}^{\dagger}$ satisfying $\Phi_K=D_Y\Phi_KD_X$. Hence, the two structures are equivalent.
\end{proof}

To make the process correspondence precise, the diagonal objects are pairs $({\mathscr H}_X,D_X)$, with ${\mathscr H}_X=\ell^2(X)$ for a nonempty finite or countable basis set $X$. Their states are the density operators
fixed by $D_X$. A morphism from $({\mathscr H}_X,D_X)$ to $({\mathscr H}_Y,D_Y)$ is a CPTP map on the corresponding trace-class spaces satisfying
\begin{eqnarray}
\Phi=D_Y\Phi D_X, \quad \id_{({\mathscr H}_X,D_X)}=D_X.
\label{eq:diagonal-process-objects}
\end{eqnarray}
Composition is ordinary channel composition. The specified identity fixes every admitted diagonal state, although it does not fix arbitrary coherent inputs. Thus the classical process category uses restricted diagonal objects; it is not an ordinary subcategory sharing the identity channels of the full quantum state spaces. Equivalently, it is a subcategory of the split-idempotent completion of the CPTP category, as detailed in Appendix~\ref{app:dephased}. The history correspondence below further restricts the sequential diagrams to prefix-preserving transitions.

\begin{corollary}[Maximal classical sector]
\label{cor:maximal-classical-sector}
Fix the history bases, and write $\diag(\Qtms_{\rm proc})$ for $\diag_{\rm hist}^{\rm pref}(\Qtms_{\rm proc})$: diagonal preparations and readouts with dephasing-fixed channels supported only on history extensions $(h_t,x_{t+1})$. HCSM is equivalent to this history-preserving sequential process sector. Symbolically,
\begin{eqnarray}
\Hcsm \cong \diag(\Qtms_{\rm proc}).
\end{eqnarray}
\end{corollary}

\begin{proof}---Theorem~\ref{thm:HCSMtoQTM} identifies every HCSM process with a diagonal quantum process built from orthogonal history states, while Theorem~\ref{thm:classical-sector} shows that every dephasing-fixed quantum channel is nothing but a stochastic kernel written in quantum notation. Restricting the recovered kernel to prefix-preserving successors gives precisely the HCSM update, so both constructions are inverse and preserve history-event probabilities and sequential composition. The identity stochastic kernel is represented by dephasing $D_X$, which is the identity on diagonal states, not on arbitrary coherent inputs.
\end{proof}

A related snapshot diagonal representation is also available at the SEnM snapshot level: whenever an ensemble-snapshot law closes as a stochastic update, the snapshot distribution may be encoded as a diagonal density matrix and evolved by a dephasing-fixed channel. The snapshot transition preserves its specified law, while the history lift additionally retains the intermediate routes; Appendix~\ref{app:diag-snap-hist} distinguishes $\diag_{\mathrm{snap}}(\Qtms_{\rm proc})$ from $\diag_{\mathrm{hist}}^{\mathrm{pref}}(\Qtms_{\rm proc})$ and shows that the former embeds as a last-coordinate subtheory of the latter.

This is the precise form of the inclusion thesis. The familiar algorithmic simulation of PTM by QTM is only the starting point. PTM and SEnM embed into HCSM, which exhausts the diagonal, commutative, positive-probability sector in its fixed-basis history-preserving form. This canonical image is a proper sector of $\Qtms_{\rm proc}$. A Hadamard sends a basis state to a coherent state outside the admitted diagonal state space. The inclusion and its strictness follow from the process correspondence and the existence of coherent quantum operations, independently of a readout-inequality violation.

\section{The quantum boundary beyond the dephased sector}\label{sec:quantum-boundary}

\subsection{Born shadows and the obstruction from composition}\label{subsec:born-comp} 

The diagonal-sector correspondence locates the full HCSM history model inside $\Qtms_{\rm proc}$. We now ask why the visible basis probabilities at this interface cannot, by themselves, preserve coherent quantum composition and inverse recovery. The key object is the Born-probability shadow of a unitary.

\begin{definition}[Born-probability shadow]
\label{def:born-shadow}
For a unitary $\hat{U}$ on a $d$-dimensional Hilbert space with computational basis $\{\ket{x}\}_{x=1}^d$, define
\begin{eqnarray}
B(\hat{U})_{yx} = \abs{\bra{y}\hat{U}\ket{x}}^2.
\end{eqnarray}
\end{definition}

The matrix $B(\hat{U})$ is doubly stochastic~\cite{Birkhoff1946}. It is a unistochastic matrix in the standard terminology~\cite{Bengtsson2005Unistochastic}. It is the transition-probability shadow obtained by preparing a basis state, applying $\hat{U}$, and measuring in the same basis. But it is not a functor from unitary composition to stochastic composition.

\begin{theorem}[Composition defect formula]
\label{thm:composition-defect}
For unitaries $\hat{U},\hat{V}$,
\begin{eqnarray}
&& B(\hat{V}\hat{U})_{zx} - [B(\hat{V})B(\hat{U})]_{zx} \nonumber \\
	&& \quad = 2\operatorname{Re}\sum_{y<y'}\hat{V}_{zy}\hat{U}_{yx}(\hat{V}_{zy'})^{*}(\hat{U}_{y'x})^{*}.
\label{eq:composition-defect-v03}
\end{eqnarray}
Consequently, $B(\hat{V}\hat{U})=B(\hat{V})B(\hat{U})$ holds if and only if the total real interference sum vanishes for each output-input pair; individual cross terms may cancel.
\end{theorem}

\begin{widetext}
\begin{proof}---Since
\begin{eqnarray}
\bra{z}\hat{V}\hat{U}\ket{x}=\sum_y\hat{V}_{zy}\hat{U}_{yx},
\end{eqnarray}
one has
\begin{eqnarray}
B(\hat{V}\hat{U})_{zx} &=& \abs{\sum_y\hat{V}_{zy}\hat{U}_{yx}}^2 = \sum_y \abs{\hat{V}_{zy}}^2 \abs{\hat{U}_{yx}}^2 + 2\operatorname{Re}\sum_{y<y'}\hat{V}_{zy}\hat{U}_{yx}(\hat{V}_{zy'})^{*}(\hat{U}_{y'x})^{*}.
\end{eqnarray}
The first term is $[B(\hat{V})B(\hat{U})]_{zx}$, proving Eq.~\eqref{eq:composition-defect-v03}.
\end{proof}
\end{widetext}

A positive stochastic shadow therefore composes probabilities, whereas a quantum circuit composes amplitudes and only then squares. That algebraic mismatch is the operational heart of interference.

\subsection{Reversibility obstruction}\label{subsec:reversibility} 

The sharpest obstruction appears when one asks a stochastic shadow to preserve inverse gates.

\begin{theorem}[Reversibility obstruction]
\label{thm:reversibility-obstruction}
Let $\hat{U}$ be unitary. If
\begin{eqnarray}
B(\hat{U}^{\dagger})B(\hat{U}) = \hat{\mathds{1}},
\end{eqnarray}
then $\hat{U}$ is monomial: there exists a permutation $\pi$ and phases $\theta_x$ such that
\begin{eqnarray}
\hat{U} = \sum_x e^{i\theta_x}\ketbra{\pi(x)}{x}.
\end{eqnarray}
Conversely, every monomial unitary satisfies $B(\hat{U}^{\dagger})B(\hat{U})=\hat{\mathds{1}}$.
\end{theorem}

\begin{proof}---Because $B(\hat{U}^{\dagger})=B(\hat{U})^T$, the condition is $B(\hat{U})^TB(\hat{U})=\hat{\mathds{1}}$. The columns of $B(\hat{U})$ are nonnegative probability vectors. Orthogonality and unit norm imply that each column has $\ell_2$ norm equal to one. But, for a nonnegative probability vector $v$, $\|v\|_2 \leq \|v\|_1 = 1$, with equality only when $v$ is one-hot. Thus, each column of $B(\hat{U})$ is one-hot. Since $B(\hat{U})$ is also doubly stochastic, those one-hot columns occupy distinct rows, so $B(\hat{U})$ is a permutation matrix. Therefore, each row and column of $\hat{U}$ contains exactly one entry of modulus one and all other entries vanish; $\hat{U}$ is monomial. The converse is immediate. A stronger stochastic-inverse version is proved in Theorem~\ref{app:thm:reversibility-strengthened} in Appendix~\ref{app:born-shadow}.
\end{proof}

Hadamard gives the simplest example. Quantum mechanically, $\hat{H}^2=\hat{\mathds{1}}$, but
\begin{eqnarray}
B(\hat{H})=\begin{pmatrix}1/2&1/2\\1/2&1/2\end{pmatrix},
\quad
B(\hat{H})^2=B(\hat{H})\neq \hat{\mathds{1}}.
\end{eqnarray}
The readout shadow forgets the phase relation that allows the second Hadamard to undo the first.

The normalized reversibility defect is
\begin{eqnarray}
R(\hat{U}) = \frac{\left\| B(\hat{U}^{\dagger})B(\hat{U}) - \hat{\mathds{1}} \right\|_F}{\sqrt d}.
\end{eqnarray}
For $\hat{H}^{\otimes k}$ this equals $R=\sqrt{1-2^{-k}}$; the analytic derivation is given in Appendix~\ref{app:born-shadow}.

\subsection{Phase activation and the phase ledger}\label{subsec:phase-ledger} 

A phase can be locally invisible in one basis and later become observable by composition. For
\begin{eqnarray}
\hat{R}_z(\theta) = \begin{pmatrix} 1 & 0 \\ 0 & e^{i\theta} \end{pmatrix},
\end{eqnarray}
one has $B(\hat{R}_z(\theta))=\hat{\mathds{1}}$ in the computational basis. Yet,
\begin{eqnarray}
B(\hat{H}\hat{R}_z(\theta)\hat{H})=
\begin{pmatrix}
\cos^2(\theta/2)	& \sin^2(\theta/2) \\
\sin^2(\theta/2)	& \cos^2(\theta/2)
\end{pmatrix},
\end{eqnarray}
whereas
\begin{eqnarray}
B(\hat{H})B(\hat{R}_z(\theta))B(\hat{H})=B(\hat{H}).
\end{eqnarray}
Thus, a phase that is absent from the local probability shadow becomes active under future composition. This is the simplest visible trace of the hidden relational information carried by coherence. The phase-dependent quantum probabilities therefore differ from the phase-independent stochastic prediction even in this single-qubit example.

We now quantify the missing information. Here $Y$ is a simulator readout available without destroying the retained record used by the inverse continuation. It is not a projective quantum measurement physically inserted before the inverse gate, which would generally destroy the coherence needed for recovery.

\begin{definition}[Phase ledger]
\label{def:phase-ledger}
Let $X$ be an input basis-state random variable with prior $\pi$, and let $Y$ be the visible readout after applying $\hat{U}$, with
\begin{eqnarray}
P(Y=y|X=x)=B(\hat{U})_{yx}.
\end{eqnarray}
If a classical stochastic simulator stores an additional hidden variable $Z$ after the step, the information in $Z$ beyond the visible readout $Y$ is the ledger information $I(X;Z|Y)$. The ideal one-step ledger cost is
\begin{eqnarray}
\Lambda_{\pi}(\hat{U}) := H_{\pi}(X|Y).
\end{eqnarray}
For the uniform prior we write simply $\Lambda(\hat{U})$.
\end{definition}

\begin{theorem}[Exact phase-ledger identity]
\label{thm:exact-phase-ledger}
Suppose a positive stochastic simulator represents a reversible unitary step $\hat{U}$ with visible readout $Y$ and additional retained information $Z$, and suppose that a later representation of $\hat{U}^{\dagger}$ recovers the original input $X$ exactly from the complete record $(Y,Z)$. Then,
\begin{eqnarray}
I(X;Z|Y)=H(X|Y).
\end{eqnarray}
Approximate recovery admits a Fano-type lower bound, given in Appendix~\ref{app:ledger}.
\end{theorem}

\begin{proof}[Proof sketch]---Exact recovery from $(Y,Z)$ implies $H(X|Y,Z)=0$. The chain-rule identity
\begin{eqnarray}
I(X;Z|Y)=H(X|Y)-H(X|Y,Z)
\end{eqnarray}
therefore reduces to $I(X;Z|Y)=H(X|Y)$. Appendix~\ref{app:ledger} gives the exact and approximate forms in detail.
\end{proof}

For the $d$-dimensional Fourier transform $\hat{F}_d$, the Born shadow is flat: $B(\hat{F}_d)_{yx}=1/d$. Under a uniform prior, $X$ and $Y$ are independent, so $\Lambda(\hat{F}_d)=\log_2 d$. For $\hat{H}^{\otimes k}\otimes \hat{\mathds{1}}^{\otimes(n-k)}$, the ledger cost is $k$ bits. The point is not that quantum theory destroys this information. Quantum theory stores it in phase and coherence, so the inverse can recover it. A positive stochastic representation must preserve the information in its retained record; HCSM can do so by keeping the input in its history. The identity fixes conditional information and lower-bounds $H(Z|Y)$, rather than a worst-case memory length.

The conditional-information identity itself applies equally to a classical
randomizing channel with the same input--output law when exact recovery is
required. Its role here is to quantify the information needed to reproduce
a specified coherent continuation: $\hat{U}^{\dagger}$ distinguishes the
orthogonal conditional outputs $\hat{U}\ket{x}$ even when their visible
basis distributions coincide. The ledger is consequently tied to the
input ensemble, readout basis, and recovery task, rather than being a
coherence measure of a single state.

There is nevertheless a direct quantitative link to coherence for the
uniform basis ensemble. Define
$\hat{\rho}_x=\hat{U}\ketbra{x}{x}\hat{U}^{\dagger}$ and let
$C_{\rm rel}(\hat{\rho})=S(D\hat{\rho})-S(\hat{\rho})$ be the
relative-entropy coherence in the readout basis~\cite{Baumgratz2014Coherence}, where
$S$ is the von Neumann entropy in bits. In finite dimension $d$,
\begin{eqnarray}
\Lambda(\hat{U})
&=&\frac{1}{d}\sum_x C_{\rm rel}(\hat{\rho}_x).
\label{eq:ledger-generated-coherence}
\end{eqnarray}
This is the average coherence generated from the labelled basis inputs,
not the coherence of their averaged density operator. The inverse can
convert the resulting coherent distinctions back into distinguishable
basis labels; a positive simulator must preserve their predictive content
in its retained record. Appendix~\ref{app:ledger} derives
Eq.~\eqref{eq:ledger-generated-coherence} and its prior-dependent extension.

\subsection{Universal exact simulation and no-free classicalization}\label{subsec:tomographic} 

The phase-ledger identity in Theorem~\ref{thm:exact-phase-ledger} concerns recovery of a basis input in one step. For general circuits, the retained stochastic representation must support all allowed future continuations. Here the relevant representation is its probability law, not necessarily an individual sampled configuration.

\begin{definition}[Exact compositional HCSM simulation of quantum circuits]
\label{def:exact-compositional-hcsm}
For an allowed initial preparation $\hat{\rho}_0$, after a circuit prefix $\hat{C}$, let $\mu_C$ be the simulator's law on its complete retained histories $h$, including all information available to subsequent updates. Each continuation $\hat{D}$ acts by causal stochastic kernels depending only on $\hat{D}$ and that retained history, and induces normalized response probabilities $\xi_{D,y}(h)\geq0$, with $\sum_y\xi_{D,y}(h)=1$. The simulation is exact and compositional if these kernels compose consistently on reachable laws and
\begin{eqnarray}
&& \sum_h \mu_C(h)\xi_{D,y}(h) = \tr\left[\ketbra{y}{y}\hat{D}\hat{\rho}_C\hat{D}^{\dagger}\right], \nonumber \\
&& \hat{\rho}_C=\hat{C}\hat{\rho}_0\hat{C}^{\dagger},
\end{eqnarray}
for every allowed prefix, continuation, and output. No additional prefix information may be supplied outside the retained representation. Two laws are operationally equivalent if all their averaged continuation responses agree. Quantum process identities are required on these equivalence classes, without requiring redundant histories to be erased.
\end{definition}

\begin{theorem}[Tomographic sufficiency]
\label{thm:tomographic-sufficiency}
If an HCSM exactly and compositionally simulates a tomographically complete family of future continuations, then its post-prefix law is tomographically sufficient for the reachable quantum state. In particular, two distinct reachable density matrices cannot have operationally equivalent simulator laws. All unitary continuations followed by computational-basis readout form such a family.
\end{theorem}

\begin{proof}
Operationally equivalent laws give the same averaged responses for every continuation. Exactness would therefore imply, for their represented states $\hat{\rho}$ and $\hat{\sigma}$,
\begin{eqnarray}
\tr\bigl[\hat{P}\hat{\rho}\bigr] = \tr\bigl[\hat{P}\hat{\sigma}\bigr],
\quad
\hat{P}=\hat{D}^{\dagger}\ketbra{y}{y}\hat{D}.
\end{eqnarray}
For all unitary continuations these are all rank-one projectors, which separate density matrices. Hence $\hat{\rho}=\hat{\sigma}$. The same argument applies to any tomographically complete continuation family. This identifies the quantum state from the operational law; it does not require every sampled history to determine that state or the laws for different states to have disjoint support. Appendix~\ref{app:tomo} gives the detailed and finite-precision forms.
\end{proof}

\begin{corollary}[Ledger dimension]
\label{cor:ledger-dimension}
Suppose the entire pure $n$-qubit state manifold admits a continuous encoding into an $m$-dimensional manifold parameterizing simulator laws, with exact tomographically complete predictions. Then $m\geq 2^{n+1}-2$. If the encoding instead covers all mixed states continuously, then $m\geq 4^n-1$.
\end{corollary}

\begin{proof}[Proof sketch]---Tomographic sufficiency makes the encoding injective. The pure-state manifold $\mathbb{CP}^{2^n-1}$ has real dimension $2^{n+1}-2$, and the full-rank mixed states form a manifold of dimension $4^n-1$. A continuous injection of either manifold into a manifold of lower dimension is excluded by invariance of domain applied in local charts. These bounds concern parameters of a full-state encoding, not bits in a sampled history or computational complexity. A countably described circuit family need not reach the full manifold; conversely, a countable configuration space may carry uncountably many probability laws.
\end{proof}

\begin{theorem}[No-free-classicalization]
\label{thm:no-free-classicalization}
For an inverse-closed family of allowed unitary gates, with all computational-basis inputs available, exact compositional representations obey the following resource alternatives:
\begin{itemize}
\item remain within the dephased stochastic sector, where HCSM applies directly and visible Born-shadow dynamics preserves unitary inverses only for monomial gates;
\item retain additional classical information beyond visible readout probabilities: under the recovery assumptions of Theorem~\ref{thm:exact-phase-ledger} it carries the one-step ledger information, and under universal exact continuation requirements its complete probability law must be tomographically sufficient;
\item change the representation's positive stochastic composition rule, for example by using signed weights, quasi-probabilities, or complex amplitudes as the compositional weights.
\end{itemize}
These alternatives describe representational resources; they do not prohibit positive classical simulation with adequate retained information.
\end{theorem}

\begin{proof}[Proof sketch]---Theorem~\ref{thm:classical-sector} identifies the fixed-basis dephased sector with stochastic dynamics. Theorem~\ref{thm:reversibility-obstruction} excludes non-monomial inverse pairs from the corresponding visible Born-shadow representation. An exact positive simulator extending beyond that representation must consequently retain information not determined by the visible probabilities. The ledger identity quantifies this information when a common retained classical state supports both the specified readout and inverse recovery; tomographic sufficiency characterizes its operational content when all separating continuations are required. Without such information, exact general simulation cannot retain the same positive stochastic representation and composition rule. Allowing other weights changes that rule, but is not by itself a sufficient construction. Appendix~\ref{app:no-free} develops these alternatives.
\end{proof}

Theorem~\ref{thm:no-free-classicalization} gives the sharpest message of this work. It is not enough to cite the algorithmic inclusion of PTM in QTM. HCSM exhausts the positive causal history sector in its diagonal quantum realization, and an exact classical representation of coherent dynamics must preserve the missing phase-dependent predictions. Amplitudes stored as classical data and updated deterministically are one realization of the hidden-ledger regime.

\subsection{Finite-precision continuation resolution}\label{subsec:continuation-resolution}

The exact boundary becomes operationally useful when it also constrains approximate predictions. Two questions remain beyond the exact statements: how many distinct predictive descriptions survive a specified error tolerance, and how does this number depend on the future experiments that the simulator must support? We study these questions on preparations with identical current basis probabilities. Any separation then comes from future phase-sensitive predictions, rather than from distinguishability already present in the visible record. The numerical task is a finite optimization over which preparations can share one predictive description; it is not a plot of the analytic ledger or state-space dimension formulas.

Let $p_i^a$ be the output distribution of continuation $a$ after preparation $i$, and let $\mathcal A$ denote the allowed continuation set. Define the continuation distance
\begin{eqnarray}
d_{\mathcal A}(i,j) &=& \sup_{a\in\mathcal A}\operatorname{TV}(p_i^a,p_j^a), \nonumber \\
\operatorname{TV}(p,q) &=& \frac12\sum_y|p(y)-q(y)|.
\label{eq:continuation-distance}
\end{eqnarray}
The error requirement is uniform: the simulated distribution differs from $p_i^a$ by at most $\epsilon$ in total variation for every tested preparation and continuation. For a binary continuation, this is simply additive error in its success probability. For all unitary continuations followed by full computational-basis readout, $d_{\mathcal A}(i,j)=\|\hat{\rho}_i-\hat{\rho}_j\|_1/2$; for pure states it equals $\sqrt{1-|\braket{\psi_i}{\psi_j}|^2}$. This full-distribution metric differs in general from the single-outcome metric $d_{\rm meas}$ in Appendix~\ref{app:tomo}; the two universal metrics coincide for the pure states used below.

\begin{proposition}[Approximate predictive resolution and finite codebooks]
\label{prop:approximate-predictive-resolution}
Suppose a positive simulator uses retained laws $\mu_i$ and a common stochastic response rule for each continuation, with the uniform error specified above. Then
\begin{eqnarray}
\operatorname{TV}(\mu_i,\mu_j)&\geq&d_{\mathcal A}(i,j)-2\epsilon.
\label{eq:law-resolution-lower-bound}
\end{eqnarray}
For a finite preparation set, form the incompatibility graph $G_\epsilon$ with an edge $ij$ whenever $d_{\mathcal A}(i,j)>2\epsilon$. If a deterministic summary assigns each preparation to one of $M$ fixed predictive laws, then $M\geq\chi(G_\epsilon)$, where $\chi$ is the chromatic number. In the special case of finitely many binary queries and arbitrary per-query Bernoulli responses, the minimum is exactly
\begin{eqnarray}
M_{\rm bin}(\epsilon)=\chi(G_\epsilon),
\quad
b_{\rm bin}(\epsilon)=\left\lceil\log_2\chi(G_\epsilon)\right\rceil.
\label{eq:exact-binary-codebook}
\end{eqnarray}
Here $b_{\rm bin}$ is the fixed-length index size of this deterministic codebook, not a memory bound for arbitrary stochastic encoders.
\end{proposition}

\begin{proof}---A common stochastic continuation contracts total variation between retained laws. Applying the triangle inequality to the two approximation errors gives Eq.~\eqref{eq:law-resolution-lower-bound}. Preparations assigned the same law therefore cannot be adjacent in $G_\epsilon$, so every deterministic codebook defines a coloring. For binary queries, every color class has probability range at most $2\epsilon$ in each coordinate. Assigning the midpoint of that coordinate's minimum and maximum gives a valid Bernoulli response within $\epsilon$ of every class member. Thus a coloring also constructs a codebook, proving Eq.~\eqref{eq:exact-binary-codebook}. This decoder is not required to be a quantum state or a compositional simulator. Appendix~\ref{app:finite-codebook} gives the metric and certificate details.
\end{proof}

\paragraph{Experiment 1: finite continuation banks versus universal resolution.}
We prepare phase-only states on an open chain,
\begin{eqnarray}
\ket{\psi(\boldsymbol\theta)}&=&\hat{V}(\boldsymbol\theta)\ket{+}^{\otimes n},
\nonumber\\
\hat{V}(\boldsymbol\theta)&=&\exp\!\left[\frac{i}{2}\left(
\sum_{j=1}^{n}\theta_j\hat{Z}_j+
\sum_{j=1}^{n-1}\theta_{n+j}\hat{Z}_j\hat{Z}_{j+1}\right)\right].
\label{eq:numerical-phase-family}
\end{eqnarray}
Here, $\hat{Z}_j$ is Pauli $Z$ on qubit $j$. Every preparation has exactly the same computational-basis probabilities $2^{-n}$, although its future statistics depend on the phases. For each $n\in\{3,4,6,8\}$ and phase width $\sigma\in\{0.15,0.45\}$, we generate $N=384$ preparations with independent $\theta_j\sim\mathcal N(0,\sigma^2)$, using three fixed seeds. This gives 24 ensembles and 9,216 preparation instances. Each ensemble has a separately seeded bank of 12 inverse-prefix continuations $\hat{H}^{\otimes n}\hat{V}(\boldsymbol\alpha_a)^{\dagger}$, with angles drawn from the same distribution. We compare their full-output continuation distance with the universal trace distance and the identically zero present-readout distance. The error scan is $\epsilon=0.05,0.10,\ldots,0.45$.

For every distance matrix and tolerance, a multi-start clique search gives a verified lower bound $L$ on the required number of deterministic predictive labels. A greedy radius-$\epsilon$ cover by sampled physical states gives a constructive upper bound $U$: the selected representative predicts every assigned preparation's tested distributions within tolerance. These are finite preparation-set certificates, not assumptions that the search finds the chromatic number. All selected clique pairs and cover assignments are checked independently of the search, with a $10^{-10}$ safety margin. The saved data also include nested banks of one and four continuations. Thus the experiment compares realizable restricted readouts with arbitrary future readouts while keeping the preparations and their present probabilities fixed.

\begin{figure*}[t]
\centering
\includegraphics[width=0.80\linewidth]{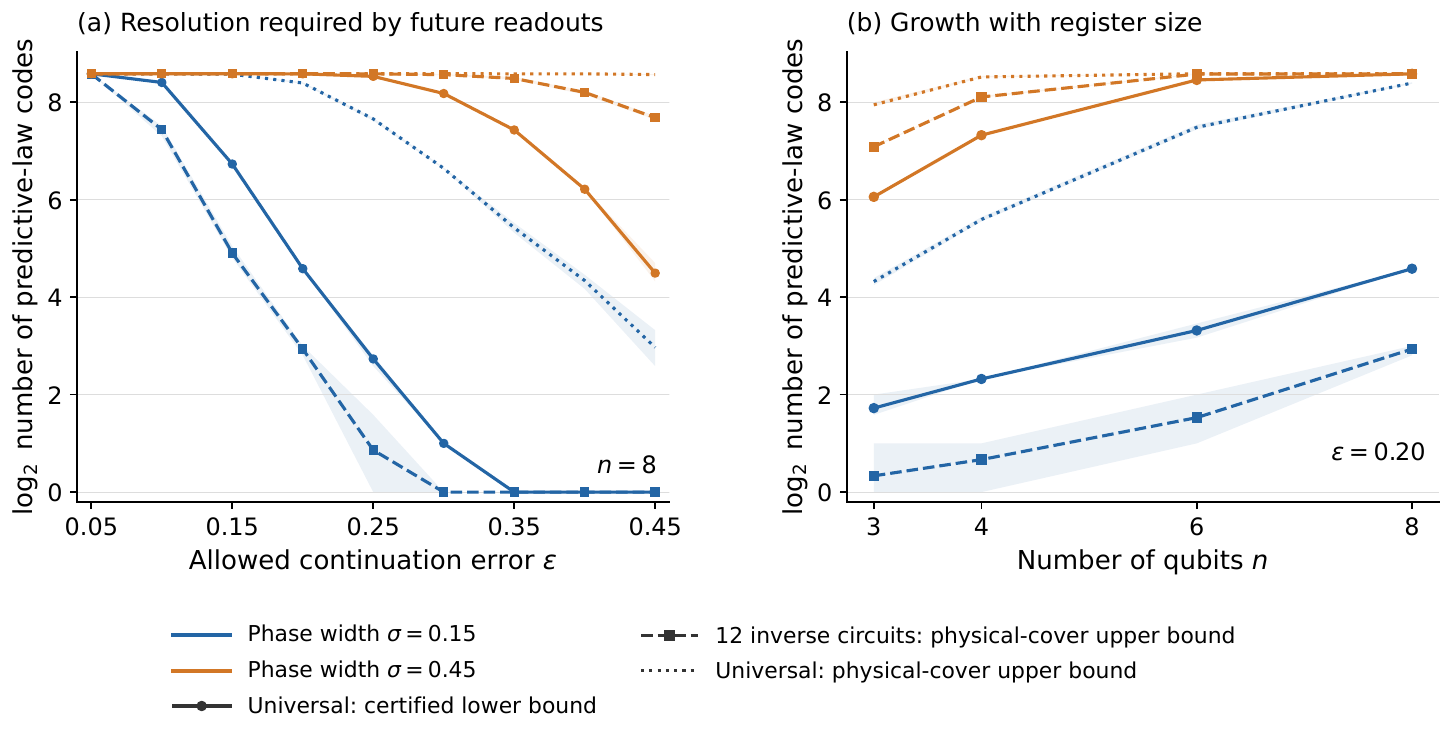}
\caption{Predictive resolution of phase-only preparations with identical present basis probabilities. Panel (a) varies the tolerance for eight qubits; panel (b) varies the register size at $\epsilon=0.20$. Solid curves certify universal-continuation lower bounds by cliques; dashed curves give constructive upper bounds for the 12-continuation bank, and dotted curves give physical-center upper bounds for universal continuations. Colors distinguish the two phase widths. Vertical coordinates are base-two logarithms of label counts; lines give their three-seed means and bands their ranges. A bank upper bound below a universal lower bound establishes an actual finite-task resolution gap, unlike a comparison of two heuristic lower bounds. All bounds concern deterministic predictive codebooks on 384 preparations, not general stochastic-encoder memory or complete compositional simulators.}
\label{fig:predictive-resolution}
\end{figure*}

Fig.~\ref{fig:predictive-resolution} shows a certified distinction between restricted and universal futures. At $n=8$, $\sigma=0.15$, and $\epsilon=0.20$, the three finite banks admit codebooks of at most $8,7,8$ labels, respectively, whereas every universal-continuation codebook requires at least 24 labels for each of those same preparation sets. Thus at most three index bits suffice for these restricted tasks, while at least five are required for universal predictions at the same tolerance. At $\epsilon=0.15$, the corresponding finite-bank upper bounds are $33,29,28$, below the universal lower bounds $108,104,107$. These are genuine gaps between a constructive upper certificate and a lower certificate, not differences between heuristics. With the broader phase width $\sigma=0.45$ and $\epsilon=0.20$, the universal lower certificates increase from $66$--$68$ labels at three qubits to $383$--$384$ at eight qubits. The latter approach the finite sample limit $N=384$ and cannot establish an asymptotic scaling law. In particular, the phase family has only $2n-1$ parameters; these data do not imply exponential memory or a complexity-class separation.

\paragraph{Experiment 2: exact codebooks and unseen continuations.}
To resolve the optimization gap itself, we next use 32 two-qubit phase preparations and binary continuation queries. Preparations are $\exp[i(a\hat{Z}_1+b\hat{Z}_2+c\hat{Z}_1\hat{Z}_2)]\ket{++}$, with $a,b,c$ independently uniform on $[0,2\pi)$. Each query projects onto an independently drawn state of the same family; it is implemented by its inverse preparation circuit followed by the binary event $00$. Four fixed seeds define independent instances, each with a nested bank of $4,16,64$ queries and tolerances $\epsilon\in\{0.05,0.10,0.15,0.20,0.25\}$. Eq.~\eqref{eq:exact-binary-codebook} now applies directly. We construct the incompatibility graph, find clique lower bounds and DSATUR colorings, and use complete colorability search if their bounds do not match. Every reported optimum is accompanied by an explicit clique and coloring, and every midpoint decoder is checked against all design probabilities.

All 60 finite-graph optima are certified: 57 have matching clique and coloring witnesses, and three require complete colorability searches to exclude one fewer color. Those three infeasibility results are independently confirmed by a mixed-integer formulation. The initial greedy colorings happen to be optimal in these instances; no improvement over that heuristic is claimed. At $\epsilon=0.20$, the minimum label counts range over $6$--$8$, $16$--$19$, and $22$--$25$ for the respective query-bank sizes, although the present basis distribution is always $(1/4,1/4,1/4,1/4)$. Their means are $7$, $17.5$, and $23.25$. The experiment thereby measures the actual information-resolution demand of the specified future interface rather than estimating it from a dimension formula.

\begin{figure*}[t]
\centering
\includegraphics[width=0.80\linewidth]{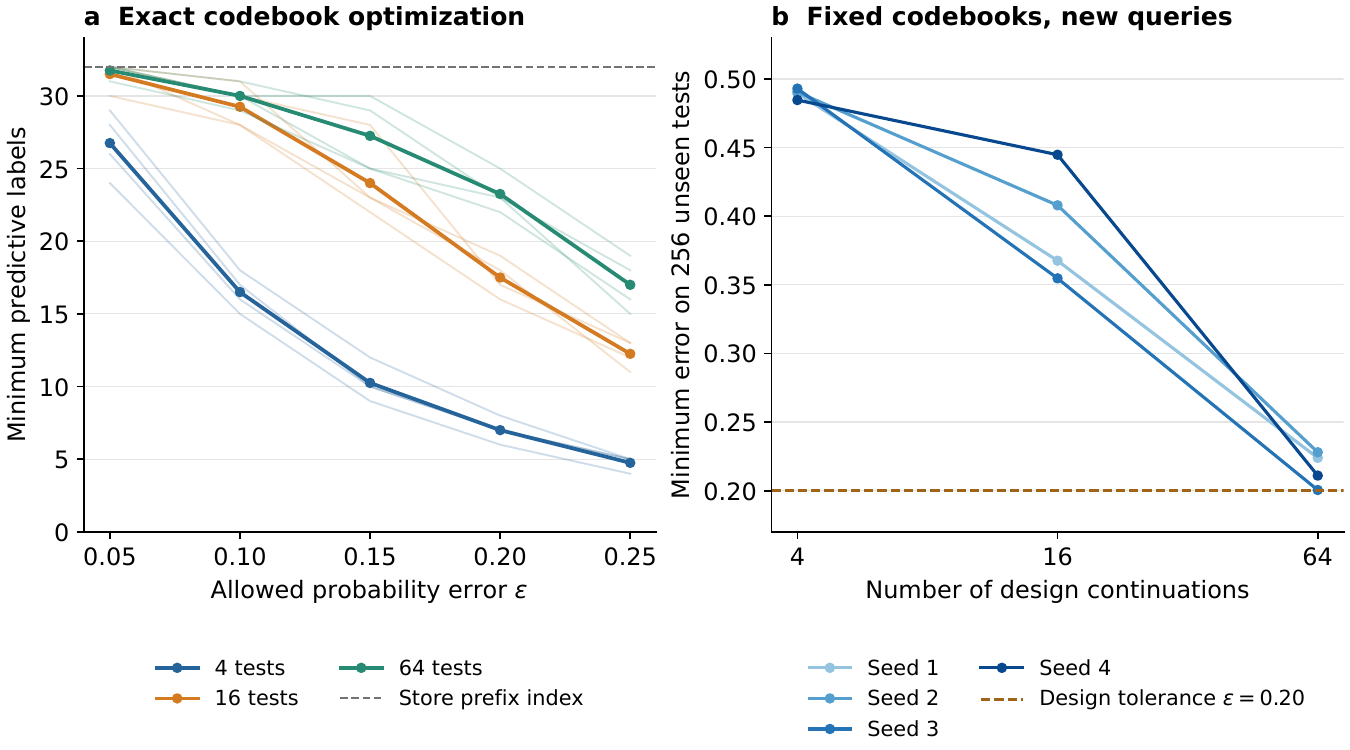}
\caption{Exact finite-query codebooks and their continuation dependence. (a) Minimum deterministic predictive labels for 32 phase-only preparations, using nested banks of binary continuation queries. Thin lines show four seeds and thick lines their means. Matching clique/coloring witnesses certify 57 instances; complete searches certify the remaining three. The dashed reference retains the full preparation index. (b) Minimum possible error on 256 new queries after fixing each encoding optimized at $\epsilon=0.20$. The new-query decoder is optimized coordinatewise, but the labels cannot be changed. The excess above the design tolerance identifies predictive distinctions discarded by the selected finite-bank encodings. This is an audit of the chosen colorings, not a bound over all optimal colorings or arbitrary stochastic encoders.}
\label{fig:exact-codebook}
\end{figure*}

We then freeze each chosen optimal encoding and challenge it with 256 independently generated, previously unused queries. To avoid blaming a particular fitted decoder, we allow the best possible new Bernoulli response for every new query. Its unavoidable worst-case error is one half of the maximum within-label probability range, as proved in Eq.~\eqref{eq:heldout-extension-error}. At design tolerance $\epsilon=0.20$, these best extension errors are $0.4846$--$0.4930$, $0.3547$--$0.4448$, and $0.2005$--$0.2280$ for encodings designed with $4,16,64$ queries, respectively (Fig.~\ref{fig:exact-codebook}). Thus even the selected finite-bank optima need not support the same tolerance for new futures. This audit concerns those particular optimal encodings; it does not optimize generalization over all equally small colorings.

Both experiments use state-vector probabilities evaluated numerically rather than shot estimates, so they isolate representational error from statistical estimation error. In the first experiment, all clique and cover witnesses pass the specified numerical checks; the smallest accepted nontrivial clique margin is $6.77\times10^{-7}$. In the second, an independent high-precision cosine expansion of all 8,192 design probabilities agrees with the double-precision calculation within $1.53\times10^{-15}$, and all 29,760 threshold comparisons agree, with minimum threshold gap $7.41\times10^{-6}$. These are numerical stability checks, not interval-arithmetic proofs; the combinatorial witnesses and complete searches certify the resulting finite graphs. Separate scripts, fixed seeds, probability matrices, assignments, and validation records accompany each experiment.

The combined result makes the approximate boundary concrete. Complete classical histories remain a valid representation, but reducing their predictive description is constrained by the future distinctions one promises to preserve. A small codebook can suffice for a restricted interface while failing for coherent continuations outside it. Approximation relaxes the information requirement quantitatively; it does not make the requirement disappear.

\section{Discussion}\label{sec:discussion}

\begin{figure*}[t]
\centering
\includegraphics[width=0.80\linewidth]{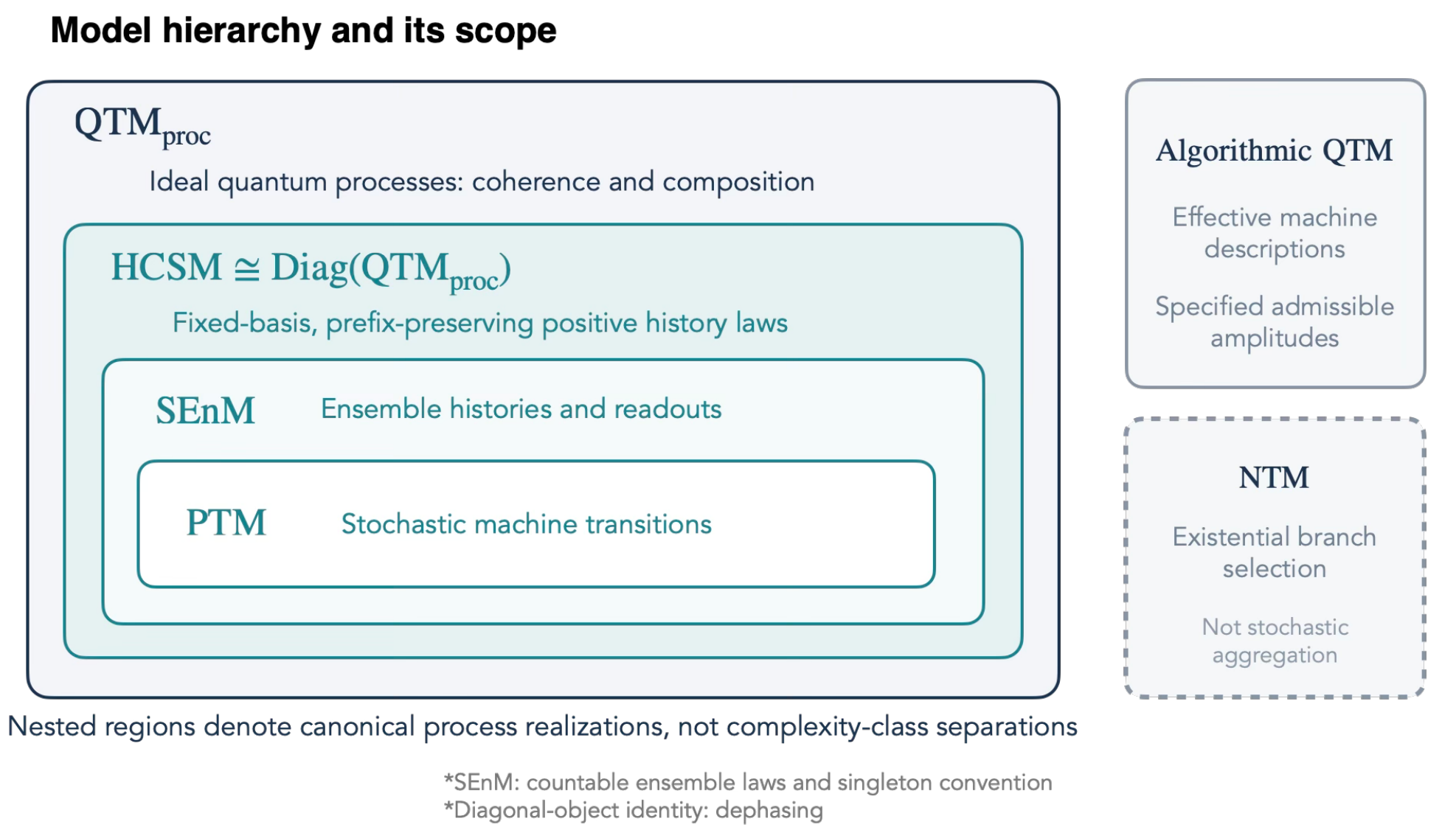}
\caption{Model hierarchy used in this work. PTM and SEnM are contained in HCSM as process models. HCSM is then embedded faithfully into the fixed-basis, history-preserving diagonal sector of ideal quantum processes $\Qtms_{\rm proc}$. The notation $\Qtms_{\rm proc}$ is distinct from the algorithmic QTM model, as specified in Sec.~\ref{subsec:process-sim}; $\diag(\Qtms_{\rm proc})$ denotes the history sector defined in Corollary~\ref{cor:maximal-classical-sector}. SEnM uses the countability and singleton conventions of Theorem~\ref{thm:SEnMinHCSM}. The NTM is drawn separately because its branch semantics is existential rather than affine in probability mass. Thus, HCSM is deliberately strong as a stochastic model, but it is not a nondeterministic branch selector.}
\label{fig:hierarchy}
\end{figure*}

The main conclusion is summarized in Fig.~\ref{fig:hierarchy}:
\begin{widetext}
\begin{eqnarray}
\Ptms \subseteq \Senm \subseteq \Hcsm \cong \diag(\Qtms_{\rm proc}) \subsetneq \Qtms_{\rm proc},
\label{eq:model-hierarchy}
\end{eqnarray}
\end{widetext}
where the inclusions concern process semantics rather than complexity-class separations. The last inclusion compares canonical state and channel descriptions. It is not an identity-preserving inclusion of categories with full quantum state spaces, since the diagonal objects have the identities specified in Eq.~\eqref{eq:diagonal-process-objects}. The final strictness is a statement about process structure, not a claim that $\mathsf{BPP} \neq \mathsf{BQP}$. Here $\diag(\Qtms_{\rm proc})$ is the complete positive stochastic sector for the fixed history bases and prefix-preserving transitions. A Hadamard sends a basis state to a coherent state outside that sector and witnesses the proper inclusion. General quantum computation uses noncommutative, coherence-bearing transformations outside it.

The result also clarifies the relation to stochastic-ensemble approaches. SEnM already allows substantial stochastic power, including non-Markovian dependence on ensemble history. Its ensemble-snapshot readout, however, leaves open what changes when complete classical routes are retained explicitly. HCSM makes those histories part of the model. It still has a global joint probability over histories and readouts, and therefore obeys the consistency conditions of ordinary probability theory, including readout inequalities of the type used in Ref.~\cite{Bang2018}. Appendix~\ref{app:senm-vs-hcsm} gives the distinction between SEnM's snapshot-ensemble semantics and HCSM's route-resolved path-measure semantics. The distinction concerns the chosen description, since an ensemble copy can itself be enlarged to retain its history. As explained in Sec.~\ref{subsec:maximality}, a readout-inequality violation tests whether the selected statistics share one joint law. The final strict inclusion in Eq.~\eqref{eq:model-hierarchy} instead follows from the diagonal-sector characterization. HCSM fills that classical sector, while coherent quantum processes extend beyond it. An enlarged classical encoding may reproduce their predictions by retaining additional quantum predictive information. The phase ledger quantifies the recovery requirement for such encodings.

HCSM already permits maximal classical branching, but it has no native representation of relational phase among branches. Quantum computation can hide a phase from one readout and later convert it into an observable probability through a change of basis. A positive stochastic model can reproduce this behavior only by storing information absent from the visible readout shadow. If the model must support all reachable circuits and unitary continuations, its hidden-state probability law becomes tomographically equivalent to the quantum state itself.

These conclusions have several limits. We do not prove a separation of complexity classes or rule out slow classical numerical simulation of finite quantum circuits. Corollary~\ref{cor:ledger-dimension} assumes a continuous encoding of a full state manifold into finite-dimensional parameters of the predictive law. Its dimension bound does not apply to one sampled history or directly count stored bits, and a countable gate orbit alone is insufficient to invoke it. The structural theorems also concern exact simulation. Section~\ref{subsec:continuation-resolution} adds quantitative finite-precision resolution bounds and numerical codebook tests. The approximate inverse-recovery bound follows from the phase-ledger identity in Theorem~\ref{thm:exact-phase-ledger} and Fano's inequality in Appendix~\ref{app:ledger}. The numerical codebook index bounds assume deterministic encoding and do not extend automatically to arbitrary stochastic encoders. Exact finite-gate implementations require compatible initial preparations and transition amplitudes. Arbitrary real countable kernels are covered by ideal quantum process semantics, but they do not automatically yield a finite-description algorithm.

The slogan $\Ptms\subseteq\Qtms$ is mathematically correct, but it leaves the boundary unspecified. Our result identifies it. Classical stochastic histories fill the fixed-basis dephased history sector of quantum computation. To represent coherent dynamics beyond visible probability closure, a model must retain the quantum predictive information in a hidden ledger or change the positive stochastic composition rule. More branching is not the missing resource. HCSM already supplies it. The remaining resource is the coherent manipulation of recoverable relations among branches.

\begin{acknowledgments}
This work was supported by the Korean ARPA-H Project through the Korea Health Industry Development Institute (KHIDI), funded by the Ministry of Health \& Welfare, Republic of Korea (RS-2025-25456722); the Ministry of Trade, Industry and Resources (MOTIR), Korea, under the project ``Industrial Technology Infrastructure Program'' (RS-2024-00466693). J.R. was supported by Korea Institute of Science and Technology Information (KISTI) (Grant No. K26L3M3C3) and by the National Research Foundation of Korea (NRF) (Grant No. RS-2023-NR119931). We acknowledge the Yonsei University Quantum Computing Project Group for providing support and access to the Quantum System One (Eagle Processor), which is operated at Yonsei University.
\end{acknowledgments}

{\bf Data Availability.}---The datasets generated and/or analyzed during the current study are not currently publicly available, as uses beyond non-commercial academic research may require institutional review under the policies of the authors’ affiliated institutions. The datasets are available from the corresponding author on reasonable request for academic research purposes.

\onecolumngrid
\appendix

\section{Classical normal form, embeddings, and readout consistency}\label{app:conventions}\label{app:normal-form}\label{app:ptm-senm}\label{app:readout}

This appendix collects the technical assumptions and readout details used in the main text. HCSM is a normal form for positive causal stochastic history laws; its precise diagonal quantum representation is developed in Appendix~\ref{app:diagonal-embedding}.

\subsection{Probability convention and scope of the normal form}
All finite-dimensional stochastic matrices in this manuscript use the column convention. If $p$ is a column vector, then
\begin{eqnarray}
 p'_y=\sum_x K(y|x)p_x .
\end{eqnarray}
Thus, a stochastic kernel satisfies
\begin{eqnarray}
 K(y|x) \geq 0,\quad \sum_yK(y|x)=1
\end{eqnarray}
for each input $x$. In this convention, a unitary Born shadow
\begin{eqnarray}
 B(\hat{U})_{yx}=\abs{\bra{y}\hat{U}\ket{x}}^2
\end{eqnarray}
is column-stochastic. Since $\hat{U}$ is unitary, it is also row-stochastic, hence doubly stochastic.

Definition~\ref{def:hcsm} uses countable spaces and a finite time horizon. This is enough for Turing-machine configurations under any fixed time bound. For countable spaces all probabilities are ordinary sums. We assume all configuration and readout sets are nonempty. A finite or countably supported ensemble-state process has the same description; the set of all infinite ensemble collections need not itself be countable. A standard Borel extension is possible using integrals and regular conditional kernels, but the countable Hilbert-space construction below is stated for the countable model.

Theorem~\ref{thm:hcsm-normal-form} already gives the normal-form proof. Its kernels are fixed by $P(h_t,x_{t+1})/P(h_t)$ whenever $P(h_t)>0$, and arbitrary on zero-probability histories. This freedom affects no path probability. The result is a factorization of a prescribed joint law; it does not by itself establish that separately specified intervention contexts share a single noninvasive joint law.

\begin{remark}[Meaning of maximality]
\label{app:rem:maximality}
Theorem~\ref{thm:hcsm-normal-form} does not say that HCSM is computationally omnipotent. It says that once a finite-horizon classical process is described by a positive joint probability distribution and causal conditioning, HCSM is a normal form for it. The maximality is relative to a fixed classical history sample space, causal conditioning, and affine aggregation over probability mass. Changing the sample space or the operational interface can represent additional information classically; the normal-form theorem alone does not exclude such representations. Postselection changes the update functional, while signed or complex representations change the weight structure.
\end{remark}

For a fixed time horizon, a probabilistic Turing machine has a countable configuration set $X_t$ at each time. A configuration includes the internal state, tape contents on the used region, and head position. Theorem~\ref{thm:PTMinHCSM} sets $K_{t+1}(x_{t+1}|h_t)=P_{\rm PTM}(x_{t+1}|x_t)$, so the path probability is
\begin{eqnarray}
 \mu_T(x_0,\ldots,x_T)=\mu_0(x_0)\prod_{t=0}^{T-1}P_{\rm PTM}(x_{t+1}|x_t),
\end{eqnarray}
which is exactly the PTM path measure. Hence PTM is not merely simulated at output level; its entire path law is reproduced.

\begin{proposition}[Affine branch aggregation is invariant under coarse graining]
\label{app:prop:affine-coarse}
Let $f:{\mathcal H}_T \to Y$ be any coarse-graining of final histories. For any event $E \subseteq Y$, the probability assigned by HCSM,
\begin{eqnarray}
{\mathcal A}_E(\mu_T)=\sum_{h_T:f(h_T) \in E}\mu_T(h_T),
\end{eqnarray}
is affine in $\mu_T$.
\end{proposition}

\begin{proof}---For distributions $\mu,\nu$ and $0\leq\lambda\leq 1$,
\begin{eqnarray}
{\mathcal A}_E(\lambda\mu + (1-\lambda)\nu) &=& \sum_{h:f(h)\in E}\bigl(\lambda\mu(h)+(1-\lambda)\nu(h)\bigr) \nonumber \\
	&=& \lambda\mathcal A_E(\mu) + (1-\lambda)\mathcal A_E(\nu).
\end{eqnarray}
\end{proof}

This is the same affine principle used in Theorem~\ref{thm:hcsm-not-ntm}: retaining histories does not grant existential branch selection. Postselection, when defined, instead uses
\begin{eqnarray}
\mu(h)\mapsto \mu(h | S)=\frac{\mu(h)\mathbf 1_{h \in S}}{\mu(S)}, \quad \mu(S)>0,
\end{eqnarray}
which is generally non-affine in the original distribution. Conditioning is legitimate for an external analyst, but HCSM does not grant it as a free internal operation that promotes a rare successful event to certainty. This is a distinction of branch semantics, not a computability separation from an NTM.

\subsection{History expansion and stochastic readouts}
HCSM is generally non-Markovian when viewed on the visible configuration $x_t$, because $K_{t+1}$ may depend on the full $h_t$. However, it is Markovian on the expanded history state.

\begin{proposition}[History Markovianization]
\label{app:prop:history-markov}
Every HCSM is a Markov chain on the history space $\mathcal{H}_t$ with time-dependent transition kernel
\begin{eqnarray}
\widehat{K}_{t+1}(h_{t+1} | h_t) =
\begin{cases}
 K_{t+1}(x_{t+1} | h_t), & h_{t+1}=(h_t,x_{t+1}), \\
 0, & \text{otherwise.}
\end{cases}
\end{eqnarray}
\end{proposition}

\begin{proof}---The update rule is
\begin{eqnarray}
 \mu_{t+1}(h_{t+1}) = \sum_{h_t}\widehat{K}_{t+1}(h_{t+1} | h_t)\mu_t(h_t).
\end{eqnarray}
Because only the unique predecessor of $h_{t+1}$ contributes, this is exactly the HCSM recursion
\begin{eqnarray}
 \mu_{t+1}(h_t,x_{t+1})=\mu_t(h_t)K_{t+1}(x_{t+1} | h_t).
\end{eqnarray}
\end{proof}

This observation is useful because it clarifies the phrase ``maximal positive causal history model.'' HCSM is the ordinary Markov description after the state variable has been enlarged to include the entire past.

The main text uses readout maps $R_t:{\mathcal H}_t\to M_t$. It is sometimes convenient, especially when comparing with SEnM, to allow a readout kernel
\begin{eqnarray}
Q_t(m_t | h_t) \geq 0, \quad \sum_{m_t}Q_t(m_t | h_t)=1.
\end{eqnarray}
For countable readout sets, this does not strengthen the model.

\begin{proposition}[Determinization of stochastic readouts]
\label{app:prop:readout-determinization}
Every countable HCSM equipped with countable stochastic readouts and a specified causal joint law for their sampling can be converted into a countable HCSM with deterministic readout maps. The conversion preserves the joint configuration and readout law.
\end{proposition}

\begin{proof}---Enlarge the configuration to $\widetilde{x}_t=(x_t,m_t)$, which remains countable. For conditionally independent readout noises, initialize and update by
\begin{eqnarray}
\widetilde\mu_0(x_0,m_0)&=&\mu_0(x_0)Q_0(m_0|x_0),\\
\widetilde K_{t+1}(x_{t+1},m_{t+1}|\widetilde h_t)&=&K_{t+1}(x_{t+1}|h_t)Q_{t+1}(m_{t+1}|h_t,x_{t+1}).
\end{eqnarray}
Then define the deterministic readout by
\begin{eqnarray}
\widetilde R_t(\widetilde h_t)=m_t.
\end{eqnarray}
The marginal law of $(m_0,\ldots,m_T)$ is exactly the same as in the original process. For correlated readout noises, single-time kernels alone do not specify a joint law. Given such a law, factor the countable joint process $(x_t,m_t)_{t=0}^T$ causally by Theorem~\ref{thm:hcsm-normal-form}; the augmented transition then may depend on the earlier readouts. No continuous random seed is needed in the configuration space. Thus deterministic readouts are a normal-form assumption for the specified countable joint process.
\end{proof}

\subsection{Joint readout distributions and entropy inequalities}
For deterministic readouts $R_t:{\mathcal H}_t \to M_t$, the induced joint law is
\begin{eqnarray}
P(m_0,\ldots,m_T)=\sum_{h_T}\mu_T(h_T)\prod_{t=0}^T\mathbf 1\{R_t(h_t)=m_t\}.
\label{eq:app-hcsm-direct-readout-joint}
\end{eqnarray}
For stochastic readouts $Q_t(m_t | h_t)$ it becomes
\begin{eqnarray}
P(m_0,\ldots,m_T)=\sum_{h_T}\mu_T(h_T)\prod_{t=0}^TQ_t(m_t | h_t),
\end{eqnarray}
provided the readout noises are conditionally independent given the history. Without conditional independence, a specified countable joint readout law can be absorbed into the enlarged history by Proposition~\ref{app:prop:readout-determinization}.

Let $M_0,\ldots,M_T$ be finite-valued HCSM readouts, or countable readouts for which the displayed entropies are finite. Since a global joint law exists, Shannon's chain rule gives
\begin{eqnarray}
H(M_1,\ldots,M_T | M_0)=\sum_{t=1}^T H(M_t | M_0,\ldots,M_{t-1}).
\end{eqnarray}
Moreover,
\begin{eqnarray}
H(M_T | M_0)\leq H(M_1,\ldots,M_T | M_0)
\end{eqnarray}
because adding variables cannot decrease joint conditional entropy. Finally,
\begin{eqnarray}
H(M_t | M_0,\ldots,M_{t-1}) \leq H(M_t | M_{t-1})
\end{eqnarray}
by monotonicity of conditional entropy under additional conditioning. Therefore,
\begin{eqnarray}
H(M_T | M_0) \leq \sum_{t=1}^T H(M_t | M_{t-1}).
\end{eqnarray}
This is the readout inequality in Corollary~\ref{cor:hcsm-readout-inequality} and in the SEnM analysis, now stated for the history-resolved HCSM model.

The inequality can be strict at either step. The first step,
\begin{eqnarray}
 H(M_T | M_0) \leq H(M_1,\ldots,M_T | M_0),
\end{eqnarray}
is equality when $(M_1,\ldots,M_{T-1})$ is a deterministic function of $(M_0,M_T)$ up to null sets. The second step is equality for time $t\geq2$ when
\begin{eqnarray}
 I(M_t;M_0,\ldots,M_{t-2} | M_{t-1})=0.
\end{eqnarray}
Thus, the equality throughout corresponds to a special Markov-like readout situation. HCSM itself need not be Markovian; the inequality is much more general.

The logical role of the readout inequality is analogous to Bell- or Leggett-Garg-type consistency conditions~\cite{Fine1982,LeggettGarg1985}: if one can assign all readouts to a single classical joint probability space, then certain inequalities follow. Quantum processes can violate such conditions when probabilities from different measurement contexts are combined as if they came from one noninvasive joint distribution. HCSM retains the same global-joint-probability property for one specified readout process. An actual sequential quantum measurement record also has a joint law and satisfies this inequality. A violation obtained by combining different contexts tests the additional common-joint-law assumption; it does not by itself rule out a context-dependent classical simulation. This qualification is independent of copy indistinguishability.

\section{SEnM and HCSM: ensemble snapshots and history laws}\label{app:senm-vs-hcsm}

This section spells out the conceptual and mathematical difference between the displayed SEnM and HCSM descriptions. HCSM makes full classical histories available explicitly, so an argument formulated for HCSM need not rely on whether individual ensemble copies can be addressed. The resulting comparison must distinguish loss of information in snapshot marginals from the expressive power of an ensemble state that may itself encode memory.

\subsection{Ensemble-snapshot and history-resolved readout descriptions}
In the SEnM framework of Ref.~\cite{Bang2018}, the state at computational step $j$ is an ensemble collection
\begin{eqnarray}
s_j=(s_1,s_2,\ldots)_j,
\end{eqnarray}
where each $s_k$ is the state of the $k$th probabilistic-machine copy. A readout symbol is assigned to the state of a copy by a map $m(s_k) \in M$, and the ensemble readout probability at step $j$ is the empirical proportion
\begin{eqnarray}
P(m_j | s_j)=\lim_{N \to \infty}\frac{1}{N}\sum_{k=1}^{N}\delta_{m(s_k),m_j} .
\label{eq:app-senm-snapshot-readout}
\end{eqnarray}
The distinguishing convention in this comparison is that the copies are \emph{indistinguishable at the readout-route level}. One can read the snapshot distribution in Eq.~\eqref{eq:app-senm-snapshot-readout}, but one does not have an operationally addressable object of the form
\begin{eqnarray}
k \longmapsto \bigl(m_0(s_k),m_1(s_k),\ldots,m_T(s_k)\bigr)
\end{eqnarray}
for a labelled copy $k$. Consequently, the route statement
\begin{eqnarray}
m_{j-1}(s_k) \longrightarrow m_j(s_k)
\end{eqnarray}
is not a primitive observable in SEnM. SEnM may use a non-Markovian transition rule for ensemble collections,
\begin{eqnarray}
C(s_j|s_{j-1},\ldots,s_0),
\end{eqnarray}
and hence it can define a joint law over ensemble snapshots,
\begin{eqnarray}
P(s_0,\ldots,s_j) = C(s_j|s_{j-1},\ldots,s_0) \cdots C(s_1|s_0)P(s_0).
 \label{eq:app-senm-ensemble-joint}
\end{eqnarray}
But the corresponding readout joint used in the readout-inequality construction is assembled from snapshot readout probabilities,
\begin{eqnarray}
 P(m_0,\ldots,m_j) = \sum_{s_0,\ldots,s_j}P(s_0,\ldots,s_j)\prod_{\ell=0}^{j}P(m_\ell | s_\ell).
 \label{eq:app-senm-grand-readout}
\end{eqnarray}
Eq.~\eqref{eq:app-senm-grand-readout} is a legitimate classical joint distribution over readout variables, but it is not obtained by tracking the readout route of one labelled copy through time. It is a joint distribution induced by a sequence of ensemble snapshots.

Thus, the specified SEnM readout interface is a snapshot-ensemble interface. It can be very strong: it can have many copies, large ensembles, and non-Markovian dependence on the past ensemble history. Still, its observable readout structure is built from ensemble-level proportions at each time, not from labelled route records. This is a statement about that interface: an ensemble state can in principle encode past information, and non-Markovian ensemble transitions already retain an ensemble-history law. The infinite-ensemble limit in Eq.~\eqref{eq:app-senm-snapshot-readout} is assumed to exist; for a finite ensemble the corresponding normalized finite sum defines the readout, including the singleton case. The ensemble-state sums use finite or countably supported laws.

HCSM changes the primitive sample space. At time $t$ the elementary object is no longer only a current configuration or an ensemble snapshot. It is the full history
\begin{eqnarray}
h_t=(x_0,x_1,\ldots,x_t) \in {\mathcal H}_t.
\end{eqnarray}
The state of the model is a probability measure over such histories,
\begin{eqnarray}
\mu_t(h_t),
\end{eqnarray}
and the causal update is
\begin{eqnarray}
\mu_{t+1}(h_t, x_{t+1})=\mu_t(h_t)K_{t+1}(x_{t+1} | h_t).
\end{eqnarray}
A readout at time $t$ is a function of the retained history,
\begin{eqnarray}
 R_t:{\mathcal H}_t \to M_t.
\end{eqnarray}
Therefore, the entire readout sequence is a function of the final history $h_T$, and the joint law is the direct path-measure marginal in Eq.~\eqref{eq:app-hcsm-direct-readout-joint}.

This is the distinction between the two displayed descriptions. SEnM constructs readout statistics from ensemble snapshots. HCSM places probability mass directly on trajectories. In the metaphor used in the main text, SEnM counts how many copies display each readout symbol at each time, while HCSM keeps the complete ledger of possible routes.

\subsection{Snapshot information loss and the scope of the embedding}
The loss incurred by retaining only one-time readout snapshots is described by a forgetful map. Given an HCSM history law $\mu_T$, one may discard route information and keep only one-time readout snapshots,
\begin{eqnarray}
{\mathcal F}_t(\mu_T)(m)=\sum_{h_T:R_t(h_t)=m}\mu_T(h_T), \quad t=0,\ldots,T.
\label{eq:app-forgetful-snapshot-map}
\end{eqnarray}
This collapse is generally many-to-one.

\begin{proposition}[Route information is strictly stronger than snapshot information]
\label{app:prop:collapse-not-injective}
The map from a history-resolved readout law to its time-slice readout marginals is not injective. Consequently, snapshot readout data do not determine route-wise readout transitions.
\end{proposition}

\begin{proof}---Let $T=1$ and let the readout alphabet be $\{0,1\}$. Consider two probability laws on readout histories:
\begin{eqnarray}
\mu^{\mathrm{stay}}(0,0) = \frac12, \quad \mu^{\mathrm{stay}}(1,1)=\frac12, \quad \mu^{\mathrm{flip}}(0,1) = \frac12, \quad \mu^{\mathrm{flip}}(1,0)=\frac12,
\end{eqnarray}
with all other readout histories assigned probability zero. Both laws have the same one-time marginals:
\begin{eqnarray}
 P(M_0=0)=P(M_0=1)=P(M_1=0)=P(M_1=1)=\frac{1}{2}.
\end{eqnarray}
Thus, they give the same snapshot readout distributions at both times. However, their route structures are opposite. Under $\mu^{\mathrm{stay}}$,
\begin{eqnarray}
 P(M_1=M_0)=1,
\end{eqnarray}
whereas under $\mu^{\mathrm{flip}}$,
\begin{eqnarray}
 P(M_1=1-M_0)=1.
\end{eqnarray}
Therefore, the time-slice snapshots do not determine whether copies persist or flip. The missing object is exactly the route-wise joint information retained by HCSM and absent from a pure snapshot summary.
\end{proof}

Proposition~\ref{app:prop:collapse-not-injective} should not be misread as saying that SEnM has no joint probability at all. SEnM has the ensemble-snapshot joint law in Eq.~\eqref{eq:app-senm-ensemble-joint} and, under the assumptions of Ref.~\cite{Bang2018}, the induced readout joint in Eq.~\eqref{eq:app-senm-grand-readout}. The point is more specific: one-time marginals do not determine temporal correlations. This does not prove that an SEnM ensemble-history law cannot encode the same correlations. HCSM exposes route information directly by taking histories as sample points.

A finite or countably supported SEnM process in the earlier readout-inequality work is described by ensemble collections $s_j$ and conditional probabilities
\begin{eqnarray}
C(s_j | s_{j-1},\ldots,s_0),
\end{eqnarray}
which generate a joint distribution over ensemble histories. Theorem~\ref{thm:SEnMinHCSM} embeds this process by setting $X_j$ equal to the possible ensemble collections and
\begin{eqnarray}
K_j(s_j | s_0,\ldots,s_{j-1})=C(s_j | s_{j-1},\ldots,s_0).
\end{eqnarray}
If the SEnM readout at step $j$ is defined by an ensemble proportion $P(m_j | s_j)$ rather than a deterministic map, Proposition~\ref{app:prop:readout-determinization} converts it to a deterministic readout by adding the sampled readout to the configuration. Thus, the HCSM inclusion remains valid for either deterministic or stochastic readout conventions.

The embedding preserves the specified joint law over ensemble collections, including its non-Markovian dependence. Consequently, it is not correct to call every SEnM diagonal representation merely snapshot-faithful: using ensemble histories gives the same history-faithful construction as for any HCSM. The difference is that HCSM permits arbitrary countable classical configurations and exposes route-dependent readouts directly. Whether this yields a strict inclusion of model classes depends on which ensemble encodings and readout operations are admitted. The noninjective map above establishes information loss for a pure one-time summary, not an unconditional strict SEnM--HCSM separation.

The SEnM/HCSM contrast at the specified interfaces can be summarized in below Tab.~\ref{tab:senm-hcsm-comparison}.
\begin{table}[h]
\centering
\begin{adjustbox}{width=0.95\textwidth}
\begin{tabular}{lll}
\toprule
\parbox[t]{0.18\linewidth}{Aspect} & \parbox[t]{0.45\linewidth}{SEnM} & \parbox[t]{0.40\linewidth}{HCSM}\\
\midrule
\parbox[t]{0.18\linewidth}{Primitive state} & \parbox[t]{0.40\linewidth}{Ensemble snapshot $s_j=(s_1,s_2,\ldots)_j$} & \parbox[t]{0.40\linewidth}{Full history $h_t=(x_0,\ldots,x_t)$}\\[0.7em]
\parbox[t]{0.18\linewidth}{Readout object} & \parbox[t]{0.40\linewidth}{Ensemble proportion $P(m_j|s_j)$} & \parbox[t]{0.40\linewidth}{History readout $R_t(h_t)$}\\[0.7em]
\parbox[t]{0.18\linewidth}{Route tracking} & \parbox[t]{0.40\linewidth}{Labelled copy routes are not operationally primitive} & \parbox[t]{0.40\linewidth}{Route histories are the sample points}\\[0.7em]
\parbox[t]{0.18\linewidth}{Memory} & \parbox[t]{0.40\linewidth}{Non-Markovian dependence on ensemble snapshots} & \parbox[t]{0.40\linewidth}{Non-Markovian dependence on full branch histories}\\[0.7em]
\parbox[t]{0.18\linewidth}{Joint readout law} & \parbox[t]{0.40\linewidth}{Constructed from ensemble-history law and snapshot readout probabilities} & \parbox[t]{0.40\linewidth}{Direct marginal of the final history measure}\\[0.7em]
\parbox[t]{0.18\linewidth}{Relation to NTM} & \parbox[t]{0.40\linewidth}{Copy labelling was regarded as an NTM-like idealization in Ref.~\cite{Bang2018}} & \parbox[t]{0.40\linewidth}{Copy/history resolution is allowed, but existential branch selection is still forbidden}\\[0.7em]
\parbox[t]{0.18\linewidth}{Main limitation} & \parbox[t]{0.40\linewidth}{Labelled copy routes are not primitive in the specified readout interface} & \parbox[t]{0.40\linewidth}{Retains classical route histories; the weights are positive probabilities rather than phase amplitudes}\\[0.7em]
\parbox[t]{0.18\linewidth}{Role in this paper} & \parbox[t]{0.40\linewidth}{Starting point and motivation} & \parbox[t]{0.40\linewidth}{Normal form for positive causal classical history laws}\\
\bottomrule
\end{tabular}
\end{adjustbox}
\caption{SEnM and HCSM use different primitive descriptions at the displayed interfaces. SEnM is naturally read as an ensemble-snapshot model: it observes readout proportions at each ensemble state. HCSM is a path-measure model: it assigns probability mass to full histories. Explicit history resolution removes the copy-indistinguishability objection while preserving ordinary stochastic, affine branch aggregation. The table does not assert that an enriched ensemble encoding cannot store histories or that the two unrestricted process classes are strictly separated.}
\label{tab:senm-hcsm-comparison}
\end{table}

Thus, the slogan ``HCSM is stronger than SEnM'' refers here to explicit access to arbitrary classical histories at the specified interface. HCSM permits causal dependence on those histories and still aggregates branches by ordinary probability. The diagonal embedding therefore applies independently of the indistinguishable-copy convention. Its obstruction to reproducing coherent composition concerns discarding phase information at that interface; it does not prohibit a classical state from recording a description of a quantum state.

\section{Diagonal quantum representation and history-preserving processes}\label{app:diagonal-embedding}\label{app:dephased}

Throughout this appendix $\Qtms_{\rm proc}$ denotes ideal quantum-process semantics, allowing CPTP maps on the stated finite-dimensional or separable Hilbert spaces. It is distinct from the algorithmic QTM model, whose machine description and admissible amplitudes require effective specification. In countably infinite dimensions channels act on trace-class operators; the associated Heisenberg-picture maps are normal.

\subsection{Kraus construction and unitary dilation}
Let ${\mathcal H}_t$ be the HCSM history space and let ${\mathscr H}_t=\ell^2({\mathcal H}_t)$ with basis $\{\ket{h_t}\}$. The embedding maps
\begin{eqnarray}
\mu_t \longmapsto \hat{\rho}_t=\sum_{h_t}\mu_t(h_t)\ket{h_t}\bra{h_t}.
\end{eqnarray}
For a transition kernel $K_{t+1}(x | h)$ define
\begin{eqnarray}
\hat{A}_{h,x}=\sqrt{K_{t+1}(x | h)}\ket{h,x}\bra{h}.
\end{eqnarray}
The trace-preservation condition is
\begin{eqnarray}
 \sum_{h,x}\hat{A}_{h,x}^{\dagger}\hat{A}_{h,x} = \sum_{h,x}K_{t+1}(x|h)\ket{h}\bra{h} = \sum_h\ket{h}\bra{h}=\hat{\mathds{1}}.
\end{eqnarray}
For countably infinite history spaces the operator sum above converges strongly, and the Kraus sum below converges in trace norm on trace-class inputs. The diagonal density-operator series also converges in trace norm. Complete positivity follows from the Kraus form. Hence,
\begin{eqnarray}
{\mathcal E}_{t+1}(\hat{\rho})=\sum_{h,x}\hat{A}_{h,x}\hat{\rho} \hat{A}_{h,x}^{\dagger}
\end{eqnarray}
is CPTP. On embedded diagonal states,
\begin{eqnarray}
{\mathcal E}_{t+1}(\hat{\rho}_t) = \sum_{h,x}\mu_t(h)K_{t+1}(x | h)\ket{h,x}\bra{h,x} = \hat{\rho}_{t+1}.
\end{eqnarray}
Thus, every HCSM transition is exactly a diagonal quantum channel.

If one insists on unitary dynamics on a larger space, introduce an environment register with orthonormal basis $\{\ket{g_{h,x}}\}$ and define an isometry on basis states by
\begin{eqnarray}
\hat{V}_{t+1}\ket{h}\ket{0}_{E} = \sum_x\sqrt{K_{t+1}(x | h)}\ket{h,x}\ket{g_{h,x}}_E.
\end{eqnarray}
The output vectors for different $h$ are orthonormal because the pairs $(h,x)$ and the garbage labels are orthogonal. After embedding input and output into a common sufficiently enlarged Hilbert space, this isometry extends to a unitary; countably infinite spaces may require additional ancillary dimensions. Tracing out the environment gives exactly the Kraus channel above. If the initial classical distribution is purified, the entire stochastic process can be represented unitarily with an environment that stores the branch labels.

\begin{remark}[Exact probabilities and computability]
\label{app:rem:computability}
For arbitrary real stochastic kernels, the construction uses amplitudes $\sqrt{K(x|h)}$ and possibly infinitely many history basis states. It is an exact representation in the ideal quantum-process semantics $\Qtms_{\rm proc}$, not a claim of exact finite implementation by an algorithmic QTM. Computable probabilities alone do not guarantee exact synthesis over a fixed finite gate set. An algorithmic implementation additionally requires an effective description, suitable state preparation, and a specified precision and resource model. Finite-dimensional, effectively specified instances can be treated by appropriate approximation schemes; unrestricted countable kernels need further effective truncation or implementation assumptions.
\end{remark}

\begin{corollary}[Effective bounded-time PTM simulation]
\label{cor:effective-PTM-QTM}
Let a PTM have a finite machine description and fair-coin transitions. For every fixed finite input and finite time bound, its full bounded-time history law has an effectively constructible, exact quantum realization over a fixed finite gate library containing Hadamard, Toffoli, and bit-flip gates. With these compatible elementary operations, the realization is implementable by an algorithmic QTM. Thus the usual inclusion $\Ptms\subseteq\Qtms$ is recovered in this effective machine setting at every finite horizon.
\end{corollary}

\begin{proof}---Prepare each required random bit with a Hadamard gate and reversibly compute the corresponding classical transition, retaining the previous configuration, the random bits, and any computational garbage. A finite time bound gives a finite construction, and the finite PTM description determines it effectively. Record each successive configuration in a history register. Distinct random strings retain orthogonal records even when they produce the same classical history, so their amplitudes do not interfere upon discarding the auxiliary registers. Measuring the history register therefore gives exactly the PTM path probabilities. Standard reversible simulation implements the same construction with the stated compatible quantum operations.
\end{proof}

\subsection{The dephasing-fixed correspondence and its identity maps} 

The main text characterizes dephasing-fixed channels by $\Phi=D_Y\Phi D_X$, where $D_X$ and $D_Y$ are the chosen basis-dephasing maps. The correspondence with stochastic kernels is explicit as follows.

Associate to a countable classical set $X$ the Hilbert space $\ell^2(X)$ together with the distinguished basis $\{\ket{x}\}_{x \in X}$. A probability distribution $p \in \Delta(X)$ corresponds to the diagonal density matrix
\begin{eqnarray}
\hat{\rho}_p=\sum_x p_x \ketbra{x}{x}.
\end{eqnarray}
A stochastic kernel $K:Y \leftarrow X$ corresponds to the dephasing-fixed channel
\begin{eqnarray}
\Phi_K(\hat{\rho})=\sum_{x,y}K(y|x)\ketbra{y}{x}\hat{\rho}\ketbra{x}{y}.
\end{eqnarray}
This channel satisfies
\begin{eqnarray}
\Phi_K(\hat{\rho}_p)=\hat{\rho}_{Kp}.
\end{eqnarray}
Composition is preserved:
\begin{eqnarray}
\Phi_L(\Phi_K(\hat{\rho}_p)) = \hat{\rho}_{L(Kp)}=\hat{\rho}_{(LK)p}=\Phi_{LK}(\hat{\rho}_p).
\end{eqnarray}
Since both sides are dephasing-fixed channels determined by their action on basis states, $\Phi_L\circ\Phi_K=\Phi_{LK}$ on all trace-class inputs. The identity stochastic kernel corresponds to $D_X$. More formally, the objects and morphisms in Eq.~\eqref{eq:diagonal-process-objects} satisfy
\begin{eqnarray}
D_Z(\Psi \circ \Phi)D_X = \Psi \circ \Phi, \quad D_Y \circ \Phi = \Phi=\Phi \circ D_X
\label{eq:diagonal-category-identities}
\end{eqnarray}
whenever $\Phi=D_Y \Phi D_X$ and $\Psi=D_Z\Psi D_Y$. Associativity follows from channel composition. An algebraic description uses the diagonal commutative von Neumann algebra
\begin{eqnarray}
{\mathcal A}_X = \left\{\sum_x a_x\ketbra{x}{x}:(a_x) \in \ell^\infty(X) \right\}, \quad ({\mathcal A}_X)_* \simeq \ell^1(X).
\label{eq:diagonal-algebra-predual}
\end{eqnarray}
The bounded-operator sum is understood strongly. The diagonal trace-class space is identified with this predual; its positive normalized elements are the probability distributions. It is therefore distinct from the algebra ${\mathcal A}_X$ itself. Stochastic state maps are dual to normal unital completely positive maps between these commutative algebras. In the equivalent split-idempotent description, an object is a quantum space equipped with an idempotent CPTP map $P$, an arrow obeys $\Phi=Q\Phi P$, and the identity at that object is $P$. Taking $P=D_X$ gives the present diagonal objects. This makes the identity convention explicit without asserting that $D_X$ is the identity on the full quantum state space.

Let $\Phi$ be a CPTP map on trace-class operators satisfying $\Phi=D_Y\Phi D_X$. Define
\begin{eqnarray}
K(y|x)=\bra{y}\Phi(\ketbra{x}{x})\ket{y}.
\end{eqnarray}
Positivity of $\Phi(\ketbra{x}{x})$ gives $K(y|x) \geq 0$, and trace preservation gives $\sum_y K(y|x)=1$. For an arbitrary input $\hat{\rho}$,
\begin{eqnarray}
\Phi(\hat{\rho}) = D_Y \Phi D_X(\hat{\rho}) = \sum_y\left(\sum_xK(y|x)\hat{\rho}_{xx}\right)\ketbra{y}{y}.
\end{eqnarray}
For countable spaces the expansion converges in trace norm. Thus, $\Phi$ is exactly the stochastic kernel $K$ written in quantum notation, as stated in Theorem~\ref{thm:classical-sector}.

The sector is proper. The Hadamard channel
\begin{eqnarray}
{\mathcal U}_H(\hat{\rho})=\hat{H}\hat{\rho} \hat{H}^{\dagger}
\end{eqnarray}
is not dephasing-fixed, since
\begin{eqnarray}
{\mathcal U}_H(\ket{0}\bra{0})=\ketbra{+}{+}
\end{eqnarray}
contains off-diagonal terms in the computational basis. Thus,
\begin{eqnarray}
D{\mathcal U}_H D \neq {\mathcal U}_H.
\end{eqnarray}
This shows that the family of dephasing-fixed maps is proper within all quantum channels. At the snapshot level, with the diagonal-object identity convention just specified, we write
\begin{eqnarray}
 \diag_{\rm snap}(\Qtms_{\rm proc}) \subsetneq \Qtms_{\rm proc}.
\end{eqnarray}
Here the strict inclusion compares the admitted state and channel
descriptions; its categorical interpretation uses the diagonal objects
above rather than the full quantum identity morphisms.
The statement is structural, not a claim that $\mathsf{BPP} \neq \mathsf{BQP}$.

\subsection{Snapshot and prefix-preserving history representations}\label{app:diag-snap-hist}
The comparison with diagonal quantum processes depends on which classical state is retained. A closed Markov dynamics on current configurations has a snapshot-diagonal representation. Expanding that state to a full history gives a history-diagonal representation. To keep these descriptions distinct, we introduce two diagonal process families.

Let $X_t$ denote a finite or countable set of current configurations, or in the SEnM interpretation a set of possible ensemble snapshots. A snapshot-diagonal process has diagonal states on ${\mathscr X}_t = \ell^2(X_t)$ of the form
\begin{eqnarray}
\hat{\rho}_t^{\mathrm{snap}} =\sum_{x_t\in X_t}p_t(x_t)\ket{x_t}\bra{x_t},
\end{eqnarray}
and dephasing-fixed channels
\begin{eqnarray}
\Phi_{t+1}^{\mathrm{snap}} = D_{X_{t+1}}\Phi_{t+1}^{\mathrm{snap}}D_{X_t}.
\end{eqnarray}
Equivalently, each such channel is a stochastic kernel
\begin{eqnarray}
 K_{t+1}^{\mathrm{snap}}(x_{t+1} | x_t),
\end{eqnarray}
with update
\begin{eqnarray}
 p_{t+1}(x_{t+1}) = \sum_{x_t}K_{t+1}^{\mathrm{snap}}(x_{t+1} | x_t)p_t(x_t).
\end{eqnarray}
We denote this state-level diagonal sector by $\diag_{\mathrm{snap}}(\Qtms_{\rm proc})$. It is the precise sense in which an affine SEnM snapshot dynamics, when closed on the snapshot distribution itself, admits a diagonal quantum representation.

By contrast, a history-diagonal prefix-preserving process uses the history spaces
\begin{eqnarray}
{\mathcal H}_t = X_0 \times \cdots \times X_t
\end{eqnarray}
and diagonal states on ${\mathscr H}_t=\ell^2({\mathcal H}_t)$ of the form
\begin{eqnarray}
\hat{\rho}_t^{\mathrm{hist}}=\sum_{h_t \in {\mathcal H}_t}\mu_t(h_t)\ket{h_t}\bra{h_t}.
\end{eqnarray}
Its transition is a dephasing-fixed channel induced by a stochastic kernel
\begin{eqnarray}
\widetilde K_{t+1}(h_{t+1} | h_t),
\end{eqnarray}
with the prefix-preserving constraint
\begin{eqnarray}
\widetilde K_{t+1}(h_{t+1} | h_t)=0 \quad\text{unless}\quad h_{t+1}=(h_t, x_{t+1})
\end{eqnarray}
for some $x_{t+1} \in X_{t+1}$. We denote this process-level sector by $\diag_{\mathrm{hist}}^{\mathrm{pref}}(\Qtms_{\rm proc})$. The HCSM embedding in Theorem~\ref{thm:HCSMtoQTM} is exactly a correspondence of specified history processes
\begin{eqnarray}
 \Hcsm\cong_{\mathrm{proc}}\diag_{\mathrm{hist}}^{\mathrm{pref}}(\Qtms_{\rm proc}),
\end{eqnarray}
because the diagonal entries are the full history probabilities $\mu_t(h_t)$ and all cylinder-event probabilities are preserved.

The prefix condition also has an intrinsic expression in the diagonal
category. Let
$\operatorname{pr}_t:{\mathcal H}_{t+1}\to{\mathcal H}_t$ delete the final
configuration and define the diagonal prefix channel by
\begin{eqnarray}
{\mathcal P}_t(\hat{\rho})
&=&\sum_{h_{t+1}}
\bra{h_{t+1}}\hat{\rho}\ket{h_{t+1}}
\ketbra{\operatorname{pr}_t(h_{t+1})}{\operatorname{pr}_t(h_{t+1})}.
\label{eq:diagonal-prefix-channel}
\end{eqnarray}
For the diagonal channel $\Phi_{t+1}^{\rm hist}$ induced by
$\widetilde K_{t+1}$, prefix preservation is equivalent to
\begin{eqnarray}
{\mathcal P}_t\circ\Phi_{t+1}^{\rm hist}
&=&D_{{\mathcal H}_t}
=\id_{({\mathscr H}_t,D_{{\mathcal H}_t})}.
\label{eq:prefix-diagonal-retraction}
\end{eqnarray}
Evaluating this equality on $\ketbra{h_t}{h_t}$ forces the total
nonnegative probability assigned to every incompatible prefix to vanish.
Conversely, the support condition and normalization imply the equality.
At every finite horizon, the specified initial law and such a sequence of
diagonal arrows therefore give exactly the HCSM history law. Conversely,
a history law yields such a diagram through its conditional extension
kernels; choices on null histories may differ without changing any event
probability, as already allowed in the normal form. These generally
rectangular history extensions can possess a stochastic left inverse;
they do not contradict the finite square inverse obstruction in
Lemma~\ref{app:lem:stochastic-inverse}.

\begin{proposition}[Snapshot diagonal processes lift to prefix-preserving history processes]
\label{app:prop:snap-lift}
For fixed nonempty configuration sets, every snapshot-diagonal stochastic process, specified by its initial law and its kernels even on null histories, has a canonical injective lift into $\diag_{\mathrm{hist}}^{\mathrm{pref}}(\Qtms_{\rm proc})$. Its image consists exactly of the prefix-preserving history processes whose next-step law depends only on the last coordinate of the current history.
\end{proposition}

\begin{proof}---Let
\begin{eqnarray}
 \ell_t:{\mathcal H}_t \to X_t, \quad \ell_t(x_0,\ldots,x_t)=x_t
\end{eqnarray}
be the last-coordinate projection. Given a snapshot kernel $K_{t+1}^{\mathrm{snap}}(x_{t+1}|x_t)$, define a history kernel by
\begin{eqnarray}
{\widetilde K}_{t+1}\bigl((h_t,x_{t+1}) | h_t\bigr) = K_{t+1}^{\mathrm{snap}}\bigl(x_{t+1} | \ell_t(h_t)\bigr)
\end{eqnarray}
and set ${\widetilde K}_{t+1}(h_{t+1} | h_t)=0$ whenever $h_{t+1}$ is not of the form $(h_t,x_{t+1})$. This kernel is prefix-preserving by construction, and it is normalized since
\begin{eqnarray}
\sum_{h_{t+1} \in \mathcal H_{t+1}}{\widetilde K}_{t+1}(h_{t+1} | h_t)=\sum_{x_{t+1} \in X_{t+1}} K_{t+1}^{\mathrm{snap}}(x_{t+1} | \ell_t(h_t))=1.
\end{eqnarray}
The Kraus construction above therefore gives a dephasing-fixed CPTP map in the prefix-preserving history sector.

The lift preserves the original snapshot dynamics after marginalization. If $\mu_0=p_0$ and
\begin{eqnarray}
\mu_{t+1}(h_t,x_{t+1}) = \mu_t(h_t)K_{t+1}^{\mathrm{snap}}(x_{t+1}|\ell_t(h_t)),
\end{eqnarray}
then the last-coordinate marginal obeys
\begin{eqnarray}
(\ell_{t+1})_* \mu_{t+1} (x_{t+1}) &=& \sum_{h_t}\mu_t(h_t)K_{t+1}^{\mathrm{snap}}(x_{t+1}|\ell_t(h_t)) \nonumber \\
	&=& \sum_{x_t}K_{t+1}^{\mathrm{snap}}(x_{t+1}|x_t) \sum_{h_t:\ell_t(h_t)=x_t}\mu_t(h_t) \nonumber\\
	&=& \sum_{x_t}K_{t+1}^{\mathrm{snap}}(x_{t+1}|x_t)p_t(x_t) \nonumber \\
	&=& p_{t+1}(x_{t+1}).
\end{eqnarray}
Thus, evolving in the lifted history theory and then forgetting the history gives the same snapshot distribution as evolving directly in the snapshot theory.

The lift is injective. If two snapshot kernels $K$ and $K'$ have the same lifted kernel, then for any $x_t$ and $x_{t+1}$, choose any history $h_t$ with $\ell_t(h_t)=x_t$. The lifted equality gives
\begin{eqnarray}
 K(x_{t+1} | x_t) = \widetilde{K}((h_t,x_{t+1}) | h_t) = \widetilde{K}'((h_t,x_{t+1}) | h_t) = K'(x_{t+1}|x_t).
\end{eqnarray}
Therefore, $K=K'$.

Finally, the image is characterized by the last-coordinate condition
\begin{eqnarray}
\widetilde{K}_{t+1}((h_t,x)|h_t) = \widetilde{K}_{t+1}((h'_t,x)|h'_t) \quad\text{whenever}\quad \ell_t(h_t)=\ell_t(h'_t).
\end{eqnarray}
Indeed, this condition allows one to define a unique snapshot kernel by
\begin{eqnarray}
 K_{t+1}^{\mathrm{snap}}(x|x_t) := \widetilde{K}_{t+1}((h_t,x) | h_t)
\end{eqnarray}
for any $h_t$ ending in $x_t$. This proves the claim.
\end{proof}

It follows that, after identifying a snapshot process with its lifted image,
\begin{eqnarray}
\diag_{\mathrm{snap}}(\Qtms_{\rm proc}) \cong \diag_{\mathrm{last}}^{\mathrm{pref}}(\Qtms_{\rm proc}) \subseteq \diag_{\mathrm{hist}}^{\mathrm{pref}}(\Qtms_{\rm proc}),
\label{eq:app-snap-last-hist-chain}
\end{eqnarray}
where $\diag_{\mathrm{last}}^{\mathrm{pref}}(\Qtms_{\rm proc})$ denotes the last-coordinate subtheory described in Proposition~\ref{app:prop:snap-lift}. For fixed coordinate spaces, the inclusion can be strict. For example, let $X_0=X_1=X_2=\{0,1\}$ and define
\begin{eqnarray}
\widetilde{K}_2\bigl((x_0,x_1,x_2) | (x_0,x_1)\bigr) =\mathbf{1}\{x_2=x_0\}.
\end{eqnarray}
This kernel is prefix-preserving, hence belongs to $\diag_{\mathrm{hist}}^{\mathrm{pref}}(\Qtms_{\rm proc})$. It is not in the last-coordinate image, because the two histories $(0,0)$ and $(1,0)$ have the same last coordinate $x_1=0$ but force different next states $x_2=0$ and $x_2=1$.

This example separates current-state Markov kernels from history-dependent kernels on the same fixed coordinate spaces. It does not separate them after arbitrary state enlargement: replacing the current state by $h_t$ makes any HCSM Markovian, as Proposition~\ref{app:prop:history-markov} shows. The injectivity statement concerns specified kernels; if processes are identified only by a law from one initial distribution, changes on unreachable histories are invisible and must be quotiented out.

For a SEnM whose snapshot law closes as a stochastic update, the same lift applies. For the more general non-Markovian SEnM convention, take the classical state to be its ensemble history and apply the direct embedding in Theorem~\ref{thm:SEnMinHCSM}. In that case the diagonal representation preserves the full ensemble-history and readout law. Consequently, the valid hierarchy is an inclusion of the specified PTM and SEnM processes into the HCSM normal form, followed by a faithful history-diagonal representation in $\Qtms_{\rm proc}$; strictness of SEnM versus HCSM requires further restrictions on ensemble encodings and readouts.

\section{Born shadows and phase-ledger bounds}\label{app:born-shadow}

\subsection{Stochastic inverses and reversible shadows}

The interference contribution in Eq.~\eqref{eq:composition-defect-v03} is a total real sum of distinct-path products. Its vanishing does not require each cross term to vanish separately, since terms can cancel. Conversely, equal relative phases do not eliminate interference: overlapping equal-phase paths can interfere constructively. Disjoint path support is sufficient, and monomial gates guarantee compositional closure throughout their group.

The reversibility obstruction in Theorem~\ref{thm:reversibility-obstruction} can be seen as an instance of the purely classical Lemma~\ref{app:lem:stochastic-inverse}.

\begin{lemma}[Stochastic inverse lemma]
\label{app:lem:stochastic-inverse}
Let $K$ be a finite $d\times d$ column-stochastic matrix. If there exists a column-stochastic matrix $L$ such that
\begin{eqnarray}
LK=\hat{\mathds{1}},
\end{eqnarray}
then $K$ is a permutation matrix.
\end{lemma}

\begin{proof}---Let $k_x$ denote the $x$th column of $K$. The equation $LK=\hat{\mathds{1}}$ says that
\begin{eqnarray}
 L k_x = e_x,
\end{eqnarray}
where $e_x$ is the $x$th standard basis vector. Since $k_x$ is a probability vector, $Lk_x$ is a convex combination of the columns of $L$. A convex combination of probability vectors equals the extreme point $e_x$ only if every column of $L$ with positive coefficient in $k_x$ is itself $e_x$. Thus the support of $k_x$ can include only columns $j$ of $L$ equal to $e_x$.

Now take two distinct inputs $x \neq x'$. If $k_x$ and $k_{x'}$ had overlapping support at some index $j$, then the $j$th column of $L$ would have to be both $e_x$ and $e_{x'}$, impossible. Therefore the supports of the columns of $K$ are pairwise disjoint. There are $d$ nonempty pairwise disjoint column supports contained in $d$ output indices. Each support must therefore consist of exactly one output index. Hence every column of $K$ is one-hot and the one-hot positions are distinct. Therefore $K$ is a permutation matrix.
\end{proof}

The square-matrix condition is necessary: the rectangular stochastic map $K=(1/2,1/2)^T$ has stochastic left inverse $L=(1,1)$ without being a permutation. Hidden-alphabet enlargement must therefore be distinguished from an inverse on the unchanged visible alphabet. For doubly stochastic $K$, the proof is even shorter: if $K^T K=\hat{\mathds{1}}$, the columns are nonnegative probability vectors of Euclidean norm one, hence one-hot.

\begin{theorem}[Reversibility obstruction, strengthened]
\label{app:thm:reversibility-strengthened}
Let $\hat{U}$ be unitary. If the Born shadow $B(\hat{U})$ has any positive stochastic inverse, then $\hat{U}$ is monomial. In particular, if $B(\hat{U}^{\dagger})B(\hat{U})=\hat{\mathds{1}}$, then $\hat{U}$ is monomial.
\end{theorem}

\begin{proof}---If $B(\hat{U})$ has a stochastic inverse, Lemma~\ref{app:lem:stochastic-inverse} implies that $B(\hat{U})$ is a permutation matrix. Thus, for each input basis vector $x$, exactly one output $y=\pi(x)$ satisfies $|\hat{U}_{yx}|^2=1$, and all other entries in that column vanish. Unitarity then forces the nonzero output positions to be distinct across columns. Hence,
\begin{eqnarray}
 \hat{U}=\sum_x e^{i\theta_x}\ketbra{\pi(x)}{x},
\end{eqnarray}
which is monomial. The special condition $B(\hat{U}^{\dagger})B(\hat{U})=\hat{\mathds{1}}$ provides the stochastic inverse $B(\hat{U}^{\dagger})$.
\end{proof}

For $\hat{U}=\hat{H}^{\otimes k}$, the Born shadow is
\begin{eqnarray}
 B(\hat{U})=\frac{1}{2^k}J_{2^k},
\end{eqnarray}
where $J_d$ is the $d\times d$ all-ones matrix. Since $J_d^2=dJ_d$,
\begin{eqnarray}
 B(\hat{U})^TB(\hat{U})=B(\hat{U})^2=B(\hat{U}).
\end{eqnarray}
The normalized reversibility defect used in the main text is
\begin{eqnarray}
 R(\hat{U})^2 = \frac{1}{d} \left\| B(\hat{U})^T B(\hat{U}) - \hat{\mathds{1}} \right\|_F^2 = \frac{1}{d} \left\| B(\hat{U}) - \hat{\mathds{1}} \right\|_F^2 = 1-\frac{1}{d}.
\end{eqnarray}
Thus,
\begin{eqnarray}
R(\hat{H}^{\otimes k})=\sqrt{1-2^{-k}}.
\end{eqnarray}

\subsection{Exact recovery, approximate recovery, and examples}\label{app:ledger}

Let $X$ be the input basis variable, $Y$ the visible readout after a unitary step, and $Z$ the additional hidden information retained by a positive stochastic simulator after that step. The visible readout law is
\begin{eqnarray}
P(Y=y | X=x)=B(\hat{U})_{yx}.
\end{eqnarray}
The readout $Y$ is assumed to be available without destroying the simulator record used by the inverse continuation. It is not a projective quantum measurement inserted before $\hat{U}^{\dagger}$. If that continuation recovers $X$ exactly from the complete retained record $(Y,Z)$, then
\begin{eqnarray}
H(X | Y,Z)=0
\end{eqnarray}
and hence
\begin{eqnarray}
I(X;Z | Y) = H(X | Y) - H(X | Y,Z) = H(X|Y).
\end{eqnarray}
Thus, the exact ledger cost is precisely
\begin{eqnarray}
\Lambda_{\pi}(\hat{U})=H_{\pi}(X | Y)
\end{eqnarray}
under input prior $\pi$. This is the exact ledger identity in Theorem~\ref{thm:exact-phase-ledger}. It fixes the conditional information $I(X;Z|Y)$ and implies $H(Z|Y)\geq H(X|Y)$ when these entropies are finite. It is not, by itself, a worst-case memory-length bound. It extends naturally to approximate recovery via Fano's inequality~\cite{CoverThomas}.

Suppose $X$ takes values in a set of size $m\geq2$ and an estimator $\widehat X=g(Y,Z)$ has error probability $e$, with
\begin{eqnarray}
e=\Pr[\widehat{X} \neq X] \leq \varepsilon \leq 1-1/m.
\end{eqnarray}
Fano's inequality gives
\begin{eqnarray}
H(X | Y,Z) \leq h_2(e)+e\log_2(m-1) \leq h_2(\varepsilon) + \varepsilon\log_2(m-1),
\end{eqnarray}
where $h_2$ is the binary entropy. The final substitution uses monotonicity of $h_2(e)+e\log_2(m-1)$ on $0\leq e\leq1-1/m$. Therefore,
\begin{eqnarray}
 I(X;Z | Y) \geq H(X | Y) - h_2(\varepsilon) - \varepsilon\log_2(m-1).
\end{eqnarray}
This is the approximate ledger lower bound; its right side may be replaced by zero when negative. If only a looser error upper bound is known, one may use the actual $e$ in Fano's inequality or the trivial nonnegative lower bound. In the maximally mixing case, $H(X|Y)=\log_2m$.

For finite-dimensional $\hat{U}$, the coherence link in
Eq.~\eqref{eq:ledger-generated-coherence} follows from purity of each
$\hat{\rho}_x=\hat{U}\ketbra{x}{x}\hat{U}^{\dagger}$:
\begin{eqnarray}
C_{\rm rel}(\hat{\rho}_x)
&=&H\bigl(B(\hat{U})_{\cdot x}\bigr),
\nonumber\\
\sum_x\pi_x C_{\rm rel}(\hat{\rho}_x)&=&H_{\pi}(Y|X).
\label{eq:ledger-conditional-coherence}
\end{eqnarray}
The Shannon chain rule therefore gives, for any input prior $\pi$,
\begin{eqnarray}
\Lambda_{\pi}(\hat{U})
&=&\sum_x\pi_x C_{\rm rel}(\hat{\rho}_x)
+H_{\pi}(X)-H_{\pi}(Y).
\label{eq:ledger-general-prior-coherence}
\end{eqnarray}
For uniform $\pi$, bistochasticity of $B(\hat{U})$ makes $Y$ uniform too,
so $H(X)=H(Y)=\log_2d$ and Eq.~\eqref{eq:ledger-generated-coherence}
follows. The average density operator is then
$d^{-1}\sum_x\hat{\rho}_x=\hat{\mathds{1}}/d$, whose coherence is zero;
the equality concerns the labelled ensemble. For a deterministic input
label, $\Lambda_{\pi}(\hat{U})=0$ even if that output is coherent.
Consequently the ledger is not a state-coherence monotone, and a positive
value by itself is not a witness that the input--output law has no
classical realization. Its operational force comes from requiring that
the same retained positive representation also support the coherent
inverse and, in the universal setting, all allowed future continuations.
Likewise, $\Lambda(\hat{R}_z(\theta))=0$ for basis inputs does not make
$\theta$ irrelevant on coherent inputs: the sequence
$\hat{H}\hat{R}_z(\theta)\hat{H}$ in Sec.~\ref{subsec:phase-ledger}
activates precisely such phase-dependent predictions.

For a monomial unitary $\hat{U}$, $Y$ determines $X$ under any prior after applying the inverse permutation. Hence $H(X|Y)=0$ and no hidden ledger is needed. For a Fourier transform $\hat{F}_d$,
\begin{eqnarray}
B(\hat{F}_d)_{yx}=\frac1d,
\end{eqnarray}
so under the uniform prior $X$ and $Y$ are independent and
\begin{eqnarray}
\Lambda(\hat{F}_d)=H(X)=\log_2d.
\end{eqnarray}
For $\hat{H}^{\otimes k} \otimes \hat{\mathds{1}}^{\otimes(n-k)}$ under a uniform $n$-bit prior, exactly the $k$ Hadamard-active input bits are hidden by the visible readout, so
\begin{eqnarray}
\Lambda=k.
\end{eqnarray}

For a Haar-random unitary in dimension $d$, each column of $B(\hat{U})$ is Dirichlet$(1,\ldots,1)$ distributed. Under a uniform input prior, double stochasticity makes $Y$ uniform and gives $H(X|Y)=d^{-1}\sum_x H(B(\hat{U})_{\cdot x})$. The expected ledger is therefore
\begin{eqnarray}
\mathbb E[\Lambda(\hat{U})]=\frac{H_d-1}{\ln 2},
\end{eqnarray}
where $H_d=\sum_{j=1}^d j^{-1}$. This follows by differentiating $\mathbb E[p^a]=\Gamma(d)\Gamma(a+1)/\Gamma(d+a)$ at $a=1$ for one component $p$ of a Dirichlet column and summing the $d$ entropy contributions. Thus the random-unitary benchmark follows analytically, without requiring a separate histogram.

\section{Tomographic sufficiency and no-free classicalization}\label{app:tomo}

\subsection{Operational hidden states and tomography} 

Let ${\mathcal C}$ be a family of reachable quantum prefixes and ${\mathcal D}$ a family of allowed future continuations. Two quantum states $\hat{\rho}$ and $\hat{\sigma}$ are operationally equivalent relative to ${\mathcal D}$ if
\begin{eqnarray}
\tr\left[\ketbra{y}{y}\hat{D}\hat{\rho} \hat{D}^{\dagger}\right] = \tr\left[\ketbra{y}{y}\hat{D}\hat{\sigma} \hat{D}^{\dagger}\right]
\end{eqnarray}
for every $\hat{D} \in {\mathcal D}$ and every computational-basis outcome $y$. If ${\mathcal D}$ contains all unitaries, these equivalence classes are singletons.

\begin{lemma}[Rank-one projectors separate states]
\label{app:lem:projectors-separate}
If $\hat{\rho}$ and $\hat{\sigma}$ are density matrices and
\begin{eqnarray}
\tr(\hat{P}\hat{\rho})=\tr(\hat{P}\hat{\sigma})
\end{eqnarray}
for every rank-one projector $\hat{P}$, then $\hat{\rho}=\hat{\sigma}$.
\end{lemma}

\begin{proof}---Let $\hat{A}=\hat{\rho} - \hat{\sigma}$. Then $\hat{A}$ is Hermitian and $\bra{v}\hat{A}\ket{v}=0$ for every unit vector $v$. The polarization identity implies $\bra{u}\hat{A}\ket{v}=0$ for all $u,v$. Hence, $\hat{A}=0$.
\end{proof}

\begin{theorem}[Tomographic sufficiency, detailed form]
\label{app:thm:tomo-detailed}
Assume an HCSM simulator exactly predicts the computational-basis output probabilities for every future unitary continuation. If two reachable prefixes produce distinct quantum states $\hat{\rho}\neq\hat{\sigma}$, then their HCSM predictive laws must be operationally distinct. Conversely, exact simulations of the same quantum state have operationally equivalent predictive laws.
\end{theorem}

\begin{proof}---Suppose distinct $\hat{\rho}$ and $\hat{\sigma}$ were represented by operationally equivalent HCSM laws $\mu_\rho$ and $\mu_\sigma$, as defined in Definition~\ref{def:exact-compositional-hcsm}. Since these laws give the same averaged response to the chosen continuation, it would assign identical probabilities to $\hat{\rho}$ and $\hat{\sigma}$ for every future unitary $\hat{D}$ and output $y$:
\begin{eqnarray}
 \tr\left[\ketbra{y}{y}\hat{D}\hat{\rho} \hat{D}^{\dagger}\right] = \tr\left[\ketbra{y}{y}\hat{D}\hat{\sigma} \hat{D}^{\dagger}\right].
\end{eqnarray}
Equivalently,
\begin{eqnarray}
\tr\left[\hat{D}^{\dagger}\ketbra{y}{y}\hat{D}\hat{\rho}\right] = \tr\left[\hat{D}^{\dagger}\ketbra{y}{y}\hat{D}\hat{\sigma}\right].
\end{eqnarray}
As $\hat{D}$ and $y$ vary, $\hat{D}^{\dagger}\ketbra{y}{y}\hat{D}$ ranges over all rank-one projectors. Lemma~\ref{app:lem:projectors-separate} then gives $\hat{\rho}=\hat{\sigma}$, contradicting the assumption. Therefore, the predictive laws must be operationally distinct. The converse follows because exactness assigns equal future predictions to equal quantum states. No preparation-independent choice of representing law is required.
\end{proof}

\subsection{Finite precision and continuous representation dimension} 

If predictions are required only up to additive error $\delta$, then the predictive laws must distinguish sufficiently separated quantum states in the corresponding measurement distance. Let
\begin{eqnarray}
d_{\rm meas}(\hat{\rho},\hat{\sigma})=\sup_{\hat{P}\,\mathrm{rank}\,1}\abs{\tr\left[\hat{P}(\hat{\rho}-\hat{\sigma})\right]}.
\end{eqnarray}
For Hermitian $\hat{A}=\hat{\rho}-\hat{\sigma}$, the spectral theorem gives $d_{\rm meas}(\hat{\rho},\hat{\sigma})=\|\hat{A}\|_\infty$. Since $\hat{A}$ is traceless, this is at most $\|\hat{A}\|_1/2$. If a simulator uses the same predictive law for both states, each true outcome probability is within $\delta$ of their common prediction. The triangle inequality therefore gives $d_{\rm meas}(\hat{\rho},\hat{\sigma})\leq2\delta$. States separated by more than $2\delta$ cannot be merged; the factor two accounts for both approximation errors. This is the finite-precision analogue of Theorem~\ref{app:thm:tomo-detailed}.

The real dimension of the pure $n$-qubit state space is the dimension of complex projective space
\begin{eqnarray}
\mathbb{CP}^{2^n-1},
\end{eqnarray}
namely
\begin{eqnarray}
2(2^n)-2=2^{n+1}-2.
\end{eqnarray}
The real dimension of the full density-matrix space on $d=2^n$ dimensions is
\begin{eqnarray}
d^2-1=4^n-1.
\end{eqnarray}
To make the dimension statement precise, suppose a full reachable state manifold $\mathcal M$ is assigned continuously to $k$ real parameters specifying its predictive law. Exactness and tomographic sufficiency make this assignment injective. For a manifold-valued parameter space, continuity allows restriction to a source neighborhood whose image lies in one target chart; the argument then uses those local target coordinates. If $\mathcal M$ has dimension $q>k$, restrict the assignment to a local chart $O\subset\mathbb R^q$ and append $q-k$ zero coordinates to its image. Invariance of domain would make the resulting continuous injection $O\to\mathbb R^q$ have open image, although its image lies in a proper linear subspace. This contradiction proves $k\geq q$ and hence the bounds in Corollary~\ref{cor:ledger-dimension}.

The continuity and full-manifold conditions are essential. A countable exact gate alphabet with a fixed initial state generates only countably many finite circuit states; density of this orbit alone does not justify the manifold dimension bound without a continuous extension. A probability law on countably many configurations can itself vary continuously, so these are dimensions of regular predictive-law encodings, not a lower bound on coordinates in one sampled configuration, bit storage, or running time. If arbitrary pure-state configurations are used directly as labels, their continuum requires a measurable-state-space extension of the countable HCSM definition.

For example, a mixed-state ensemble may select $\ket{0}$ or $\ket{1}$ with weights $(1/2,1/2)$ or $(3/4,1/4)$. Evolving the selected pure state unitarily represents distinct density matrices with the same initial configuration support. The distinction resides in the preparation law, not in the identity of one sampled configuration. For a countable gate alphabet the exactly reached pure-state labels in this example remain countable.

\subsection{Continuation distances and finite-codebook certificates}\label{app:finite-codebook}

For the setting of Proposition~\ref{prop:approximate-predictive-resolution}, write $q_i^a(y)=\sum_h\mu_i(h)\xi_{a,y}(h)$, where the response probabilities are nonnegative and sum to one. Then
\begin{eqnarray}
\operatorname{TV}(q_i^a,q_j^a)
&\leq&\frac12\sum_{y,h}|\mu_i(h)-\mu_j(h)|\xi_{a,y}(h)
=\operatorname{TV}(\mu_i,\mu_j),
\nonumber\\
\operatorname{TV}(p_i^a,p_j^a)
&\leq&2\epsilon+\operatorname{TV}(\mu_i,\mu_j).
\end{eqnarray}
Taking the supremum over $a$ proves the law-resolution bound without assuming a deterministic encoder. A codebook restriction is a separate assumption: a preparation index selects one of $M$ complete laws. An encoder that instead changes mixture weights over a smaller sampled alphabet need not satisfy this restriction, so a packing of laws cannot be reinterpreted as an alphabet-size bound.

For all unitary continuations with full basis readout, measurement contraction bounds the distance by $\|\hat{\rho}_i-\hat{\rho}_j\|_1/2$. Choosing the readout basis to diagonalize the Hermitian difference attains this bound. For two pure states the difference has two possible nonzero eigenvalues, $\pm\sqrt{1-|\braket{\psi_i}{\psi_j}|^2}$, proving the formula used in the first experiment. Restricting the continuation family can only decrease the distance. In contrast, the universal single rank-one outcome distance is the operator norm, as defined earlier in this appendix.

A clique with $L$ vertices in $G_\epsilon$ gives $L\leq\chi(G_\epsilon)$ and can be checked directly by testing all its pairs. An explicit coloring with $U$ colors gives $\chi(G_\epsilon)\leq U$ and is checked by verifying that every edge joins different colors. Matching witnesses certify the finite graph's exact chromatic number independently of the search heuristic. For binary queries the midpoint decoder then certifies its operational interpretation in Eq.~\eqref{eq:exact-binary-codebook}. For full outcome distributions, pairwise distances at most $2\epsilon$ do not generally guarantee a common normalized decoder of radius $\epsilon$; we therefore do not use the binary equality in the first experiment. Instead, a selected physical preparation within distance $\epsilon$ of every assigned member gives a constructive center and a normalized response for all tested continuations. Covering the finite preparation set in this way yields an upper bound for its prediction codebook, but supplies no update rule for arbitrary subsequent compositions.

The error of extending a fixed binary-query encoding to a new test bank can also be evaluated exactly. For its color classes $\mathcal C_s$ and new probabilities $p_i^a(1)$, the smallest achievable worst-case error of any extension decoder is
\begin{eqnarray}
\epsilon_{\rm ext}&=&\frac12\max_{s,a}
\left[\max_{i\in\mathcal C_s}p_i^a(1)-\min_{i\in\mathcal C_s}p_i^a(1)\right].
\label{eq:heldout-extension-error}
\end{eqnarray}
Each coordinate midpoint attains the bound. This evaluates the selected encoding; it is not an optimization over all colorings with the same number of labels. Consequently, a small codebook optimal for one continuation bank need not retain enough information for a larger bank. Both experiments fix their preparation and continuation seeds before optimization and retain all tested instances in the accompanying data.

\subsection{The no-free-classicalization alternatives}\label{app:no-free}

Theorem~\ref{thm:no-free-classicalization} separates three representation regimes while retaining the original hierarchy of classical and quantum processes.
\begin{enumerate}
\item On the unchanged visible basis alphabet, basis preparation and basis readout fix the stochastic gate map to $B(\hat{U})$. An inverse-closed family therefore requires $B(\hat{U}^{\dagger})B(\hat{U})=\hat{\mathds{1}}$. The stochastic-inverse lemma implies monomiality, and monomial unitaries conversely act by permutations that compose correctly. This is the reversible probability-only part of the dephased classical sector.
\item If a positive model enlarges its retained description to represent non-monomial reversible gates, the inverse-recovery assumptions of Theorem~\ref{thm:exact-phase-ledger} give $I(X;Z|Y)=H(X|Y)$. If all future unitary continuations must be predicted, Theorem~\ref{app:thm:tomo-detailed} additionally forces the predictive-law classes to distinguish all reachable density matrices. Thus the hidden ledger can carry the quantum predictive information in a classical representation.
\item A representation that accommodates such coherent dynamics without this positive hidden-information description cannot also remain closed under ordinary stochastic updating of the same visible basis probabilities. It must change its representation rules, for example by using signed or quasiprobability weights. Complex amplitudes may serve as such non-stochastic weights; when they are instead stored as classical configuration data and updated deterministically, they belong to the preceding positive hidden-ledger regime.
\end{enumerate}

The trichotomy concerns representation resources. The canonical diagonal image of HCSM has positive history weights, but an enlarged HCSM can store a circuit description or amplitude information and reproduce coherent quantum predictions. It then satisfies the information requirement rather than eliminating it. Accordingly, the phase ledger and tomographic statement identify the cost of classicalization beyond visible probability closure; they do not prohibit exact positive classical simulation with sufficient retained information.


\begin{thebibliography}{26}%
\makeatletter
\providecommand \@ifxundefined [1]{%
 \@ifx{#1\undefined}
}%
\providecommand \@ifnum [1]{%
 \ifnum #1\expandafter \@firstoftwo
 \else \expandafter \@secondoftwo
 \fi
}%
\providecommand \@ifx [1]{%
 \ifx #1\expandafter \@firstoftwo
 \else \expandafter \@secondoftwo
 \fi
}%
\providecommand \natexlab [1]{#1}%
\providecommand \enquote  [1]{``#1''}%
\providecommand \bibnamefont  [1]{#1}%
\providecommand \bibfnamefont [1]{#1}%
\providecommand \citenamefont [1]{#1}%
\providecommand \href@noop [0]{\@secondoftwo}%
\providecommand \href [0]{\begingroup \@sanitize@url \@href}%
\providecommand \@href[1]{\@@startlink{#1}\@@href}%
\providecommand \@@href[1]{\endgroup#1\@@endlink}%
\providecommand \@sanitize@url [0]{\catcode `\\12\catcode `\$12\catcode
  `\&12\catcode `\#12\catcode `\^12\catcode `\_12\catcode `\%12\relax}%
\providecommand \@@startlink[1]{}%
\providecommand \@@endlink[0]{}%
\providecommand \url  [0]{\begingroup\@sanitize@url \@url }%
\providecommand \@url [1]{\endgroup\@href {#1}{\urlprefix }}%
\providecommand \urlprefix  [0]{URL }%
\providecommand \Eprint [0]{\href }%
\providecommand \doibase [0]{https://doi.org/}%
\providecommand \selectlanguage [0]{\@gobble}%
\providecommand \bibinfo  [0]{\@secondoftwo}%
\providecommand \bibfield  [0]{\@secondoftwo}%
\providecommand \translation [1]{[#1]}%
\providecommand \BibitemOpen [0]{}%
\providecommand \bibitemStop [0]{}%
\providecommand \bibitemNoStop [0]{.\EOS\space}%
\providecommand \EOS [0]{\spacefactor3000\relax}%
\providecommand \BibitemShut  [1]{\csname bibitem#1\endcsname}%
\let\auto@bib@innerbib\@empty
\bibitem [{\citenamefont {Feynman}(1986)}]{Feynman1986}%
  \BibitemOpen
  \bibfield  {author} {\bibinfo {author} {\bibfnamefont {R.~P.}\ \bibnamefont
  {Feynman}},\ }\href@noop {} {\bibfield  {journal} {\bibinfo  {journal}
  {Found. Phys.}\ }\textbf {\bibinfo {volume} {16}},\ \bibinfo {pages} {507}
  (\bibinfo {year} {1986})}\BibitemShut {NoStop}%
\bibitem [{\citenamefont {Grover}(1997)}]{Grover1997}%
  \BibitemOpen
  \bibfield  {author} {\bibinfo {author} {\bibfnamefont {L.~K.}\ \bibnamefont
  {Grover}},\ }\href@noop {} {\bibfield  {journal} {\bibinfo  {journal} {Phys.
  Rev. Lett.}\ }\textbf {\bibinfo {volume} {79}},\ \bibinfo {pages} {325}
  (\bibinfo {year} {1997})}\BibitemShut {NoStop}%
\bibitem [{\citenamefont {Shor}(1997)}]{Shor1997}%
  \BibitemOpen
  \bibfield  {author} {\bibinfo {author} {\bibfnamefont {P.~W.}\ \bibnamefont
  {Shor}},\ }\href@noop {} {\bibfield  {journal} {\bibinfo  {journal} {SIAM J.
  Comput.}\ }\textbf {\bibinfo {volume} {26}},\ \bibinfo {pages} {1484}
  (\bibinfo {year} {1997})}\BibitemShut {NoStop}%
\bibitem [{\citenamefont {Yao}(1993)}]{Yao1993}%
  \BibitemOpen
  \bibfield  {author} {\bibinfo {author} {\bibfnamefont {A.~C.-C.}\
  \bibnamefont {Yao}},\ }in\ \href@noop {} {\emph {\bibinfo {booktitle}
  {Proceedings of the 34th Annual Symposium on Foundations of Computer
  Science}}}\ (\bibinfo  {publisher} {IEEE},\ \bibinfo {year} {1993})\ pp.\
  \bibinfo {pages} {352--361}\BibitemShut {NoStop}%
\bibitem [{\citenamefont {Bernstein}\ and\ \citenamefont
  {Vazirani}(1993)}]{BernsteinVazirani1993}%
  \BibitemOpen
  \bibfield  {author} {\bibinfo {author} {\bibfnamefont {E.}~\bibnamefont
  {Bernstein}}\ and\ \bibinfo {author} {\bibfnamefont {U.}~\bibnamefont
  {Vazirani}},\ }in\ \href@noop {} {\emph {\bibinfo {booktitle} {Proceedings of
  the 25th Annual ACM Symposium on Theory of Computing}}}\ (\bibinfo
  {publisher} {ACM},\ \bibinfo {year} {1993})\ pp.\ \bibinfo {pages}
  {11--20}\BibitemShut {NoStop}%
\bibitem [{\citenamefont {Bang}\ \emph {et~al.}(2018)\citenamefont {Bang},
  \citenamefont {Ryu}, \citenamefont {Lee}, \citenamefont {Yee}, \citenamefont
  {Lee},\ and\ \citenamefont {Son}}]{Bang2018}%
  \BibitemOpen
  \bibfield  {author} {\bibinfo {author} {\bibfnamefont {J.}~\bibnamefont
  {Bang}}, \bibinfo {author} {\bibfnamefont {J.}~\bibnamefont {Ryu}}, \bibinfo
  {author} {\bibfnamefont {C.-W.}\ \bibnamefont {Lee}}, \bibinfo {author}
  {\bibfnamefont {K.~H.}\ \bibnamefont {Yee}}, \bibinfo {author} {\bibfnamefont
  {J.}~\bibnamefont {Lee}},\ and\ \bibinfo {author} {\bibfnamefont
  {W.}~\bibnamefont {Son}},\ }\href@noop {} {\bibfield  {journal} {\bibinfo
  {journal} {Adv. Quantum Technol.}\ }\textbf {\bibinfo {volume} {1}},\
  \bibinfo {pages} {1800037} (\bibinfo {year} {2018})}\BibitemShut {NoStop}%
\bibitem [{\citenamefont {Cover}\ and\ \citenamefont
  {Thomas}(2006)}]{CoverThomas}%
  \BibitemOpen
  \bibfield  {author} {\bibinfo {author} {\bibfnamefont {T.~M.}\ \bibnamefont
  {Cover}}\ and\ \bibinfo {author} {\bibfnamefont {J.~A.}\ \bibnamefont
  {Thomas}},\ }\href@noop {} {\emph {\bibinfo {title} {{Elements of Information
  Theory}}}},\ \bibinfo {edition} {2nd}\ ed.\ (\bibinfo  {publisher} {Wiley},\
  \bibinfo {address} {Hoboken},\ \bibinfo {year} {2006})\BibitemShut {NoStop}%
\bibitem [{\citenamefont {Coecke}\ \emph {et~al.}(2016)\citenamefont {Coecke},
  \citenamefont {Heunen},\ and\ \citenamefont
  {Kissinger}}]{Coecke2016Channels}%
  \BibitemOpen
  \bibfield  {author} {\bibinfo {author} {\bibfnamefont {B.}~\bibnamefont
  {Coecke}}, \bibinfo {author} {\bibfnamefont {C.}~\bibnamefont {Heunen}},\
  and\ \bibinfo {author} {\bibfnamefont {A.}~\bibnamefont {Kissinger}},\
  }\href@noop {} {\bibfield  {journal} {\bibinfo  {journal} {Quantum Inf.
  Process.}\ }\textbf {\bibinfo {volume} {15}},\ \bibinfo {pages} {5179}
  (\bibinfo {year} {2016})}\BibitemShut {NoStop}%
\bibitem [{\citenamefont {Chiribella}\ \emph {et~al.}(2009)\citenamefont
  {Chiribella}, \citenamefont {D'Ariano},\ and\ \citenamefont
  {Perinotti}}]{Chiribella2009Networks}%
  \BibitemOpen
  \bibfield  {author} {\bibinfo {author} {\bibfnamefont {G.}~\bibnamefont
  {Chiribella}}, \bibinfo {author} {\bibfnamefont {G.~M.}\ \bibnamefont
  {D'Ariano}},\ and\ \bibinfo {author} {\bibfnamefont {P.}~\bibnamefont
  {Perinotti}},\ }\href@noop {} {\bibfield  {journal} {\bibinfo  {journal}
  {Phys. Rev. A}\ }\textbf {\bibinfo {volume} {80}},\ \bibinfo {pages} {022339}
  (\bibinfo {year} {2009})}\BibitemShut {NoStop}%
\bibitem [{\citenamefont {Pollock}\ \emph {et~al.}(2018)\citenamefont
  {Pollock}, \citenamefont {Rodr{\'i}guez-Rosario}, \citenamefont {Frauenheim},
  \citenamefont {Paternostro},\ and\ \citenamefont
  {Modi}}]{Pollock2018Processes}%
  \BibitemOpen
  \bibfield  {author} {\bibinfo {author} {\bibfnamefont {F.~A.}\ \bibnamefont
  {Pollock}}, \bibinfo {author} {\bibfnamefont {C.}~\bibnamefont
  {Rodr{\'i}guez-Rosario}}, \bibinfo {author} {\bibfnamefont {T.}~\bibnamefont
  {Frauenheim}}, \bibinfo {author} {\bibfnamefont {M.}~\bibnamefont
  {Paternostro}},\ and\ \bibinfo {author} {\bibfnamefont {K.}~\bibnamefont
  {Modi}},\ }\href@noop {} {\bibfield  {journal} {\bibinfo  {journal} {Phys.
  Rev. A}\ }\textbf {\bibinfo {volume} {97}},\ \bibinfo {pages} {012127}
  (\bibinfo {year} {2018})}\BibitemShut {NoStop}%
\bibitem [{\citenamefont {Milz}\ \emph {et~al.}(2020)\citenamefont {Milz},
  \citenamefont {Sakuldee}, \citenamefont {Pollock},\ and\ \citenamefont
  {Modi}}]{Milz2020Extension}%
  \BibitemOpen
  \bibfield  {author} {\bibinfo {author} {\bibfnamefont {S.}~\bibnamefont
  {Milz}}, \bibinfo {author} {\bibfnamefont {F.}~\bibnamefont {Sakuldee}},
  \bibinfo {author} {\bibfnamefont {F.~A.}\ \bibnamefont {Pollock}},\ and\
  \bibinfo {author} {\bibfnamefont {K.}~\bibnamefont {Modi}},\ }\href@noop {}
  {\bibfield  {journal} {\bibinfo  {journal} {Quantum}\ }\textbf {\bibinfo
  {volume} {4}},\ \bibinfo {pages} {255} (\bibinfo {year} {2020})}\BibitemShut
  {NoStop}%
\bibitem [{\citenamefont {Baumgratz}\ \emph {et~al.}(2014)\citenamefont
  {Baumgratz}, \citenamefont {Cramer},\ and\ \citenamefont
  {Plenio}}]{Baumgratz2014Coherence}%
  \BibitemOpen
  \bibfield  {author} {\bibinfo {author} {\bibfnamefont {T.}~\bibnamefont
  {Baumgratz}}, \bibinfo {author} {\bibfnamefont {M.}~\bibnamefont {Cramer}},\
  and\ \bibinfo {author} {\bibfnamefont {M.~B.}\ \bibnamefont {Plenio}},\
  }\href@noop {} {\bibfield  {journal} {\bibinfo  {journal} {Phys. Rev. Lett.}\
  }\textbf {\bibinfo {volume} {113}},\ \bibinfo {pages} {140401} (\bibinfo
  {year} {2014})}\BibitemShut {NoStop}%
\bibitem [{\citenamefont {Chitambar}\ and\ \citenamefont
  {Gour}(2016)}]{ChitambarGour2016Comparison}%
  \BibitemOpen
  \bibfield  {author} {\bibinfo {author} {\bibfnamefont {E.}~\bibnamefont
  {Chitambar}}\ and\ \bibinfo {author} {\bibfnamefont {G.}~\bibnamefont
  {Gour}},\ }\href@noop {} {\bibfield  {journal} {\bibinfo  {journal} {Phys.
  Rev. A}\ }\textbf {\bibinfo {volume} {94}},\ \bibinfo {pages} {052336}
  (\bibinfo {year} {2016})},\ \bibinfo {note} {erratum, Phys. Rev. A
  \textbf{95}, 019902 (2017)}\BibitemShut {NoStop}%
\bibitem [{\citenamefont {Bengtsson}\ \emph {et~al.}(2005)\citenamefont
  {Bengtsson}, \citenamefont {Ericsson}, \citenamefont {Ku{\'s}}, \citenamefont
  {Tadej},\ and\ \citenamefont {{\.Z}yczkowski}}]{Bengtsson2005Unistochastic}%
  \BibitemOpen
  \bibfield  {author} {\bibinfo {author} {\bibfnamefont {I.}~\bibnamefont
  {Bengtsson}}, \bibinfo {author} {\bibfnamefont {{\AA}.}~\bibnamefont
  {Ericsson}}, \bibinfo {author} {\bibfnamefont {M.}~\bibnamefont {Ku{\'s}}},
  \bibinfo {author} {\bibfnamefont {W.}~\bibnamefont {Tadej}},\ and\ \bibinfo
  {author} {\bibfnamefont {K.}~\bibnamefont {{\.Z}yczkowski}},\ }\href@noop {}
  {\bibfield  {journal} {\bibinfo  {journal} {Commun. Math. Phys.}\ }\textbf
  {\bibinfo {volume} {259}},\ \bibinfo {pages} {307} (\bibinfo {year}
  {2005})}\BibitemShut {NoStop}%
\bibitem [{\citenamefont {Spekkens}(2008)}]{Spekkens2008Negativity}%
  \BibitemOpen
  \bibfield  {author} {\bibinfo {author} {\bibfnamefont {R.~W.}\ \bibnamefont
  {Spekkens}},\ }\href@noop {} {\bibfield  {journal} {\bibinfo  {journal}
  {Phys. Rev. Lett.}\ }\textbf {\bibinfo {volume} {101}},\ \bibinfo {pages}
  {020401} (\bibinfo {year} {2008})}\BibitemShut {NoStop}%
\bibitem [{\citenamefont {Veitch}\ \emph {et~al.}(2012)\citenamefont {Veitch},
  \citenamefont {Ferrie}, \citenamefont {Gross},\ and\ \citenamefont
  {Emerson}}]{Veitch2012Negativity}%
  \BibitemOpen
  \bibfield  {author} {\bibinfo {author} {\bibfnamefont {V.}~\bibnamefont
  {Veitch}}, \bibinfo {author} {\bibfnamefont {C.}~\bibnamefont {Ferrie}},
  \bibinfo {author} {\bibfnamefont {D.}~\bibnamefont {Gross}},\ and\ \bibinfo
  {author} {\bibfnamefont {J.}~\bibnamefont {Emerson}},\ }\href@noop {}
  {\bibfield  {journal} {\bibinfo  {journal} {New J. Phys.}\ }\textbf {\bibinfo
  {volume} {14}},\ \bibinfo {pages} {113011} (\bibinfo {year}
  {2012})}\BibitemShut {NoStop}%
\bibitem [{\citenamefont {Montina}(2011)}]{Montina2011Dimension}%
  \BibitemOpen
  \bibfield  {author} {\bibinfo {author} {\bibfnamefont {A.}~\bibnamefont
  {Montina}},\ }\href@noop {} {\bibfield  {journal} {\bibinfo  {journal} {Phys.
  Rev. A}\ }\textbf {\bibinfo {volume} {83}},\ \bibinfo {pages} {032107}
  (\bibinfo {year} {2011})}\BibitemShut {NoStop}%
\bibitem [{\citenamefont {Cobucci}\ \emph {et~al.}(2026)\citenamefont
  {Cobucci}, \citenamefont {Bernal}, \citenamefont {Renner},\ and\
  \citenamefont {Tavakoli}}]{Cobucci2026Operational}%
  \BibitemOpen
  \bibfield  {author} {\bibinfo {author} {\bibfnamefont {G.}~\bibnamefont
  {Cobucci}}, \bibinfo {author} {\bibfnamefont {A.}~\bibnamefont {Bernal}},
  \bibinfo {author} {\bibfnamefont {M.~J.}\ \bibnamefont {Renner}},\ and\
  \bibinfo {author} {\bibfnamefont {A.}~\bibnamefont {Tavakoli}},\ }\href@noop
  {} {\bibfield  {journal} {\bibinfo  {journal} {Nat. Commun.}\ }\textbf
  {\bibinfo {volume} {17}},\ \bibinfo {pages} {1104} (\bibinfo {year}
  {2026})}\BibitemShut {NoStop}%
\bibitem [{\citenamefont {Benioff}(1982)}]{Benioff1982}%
  \BibitemOpen
  \bibfield  {author} {\bibinfo {author} {\bibfnamefont {P.}~\bibnamefont
  {Benioff}},\ }\href@noop {} {\bibfield  {journal} {\bibinfo  {journal} {Phys.
  Rev. Lett.}\ }\textbf {\bibinfo {volume} {48}},\ \bibinfo {pages} {1581}
  (\bibinfo {year} {1982})}\BibitemShut {NoStop}%
\bibitem [{\citenamefont {Deutsch}(1985)}]{Deutsch1985}%
  \BibitemOpen
  \bibfield  {author} {\bibinfo {author} {\bibfnamefont {D.}~\bibnamefont
  {Deutsch}},\ }\href@noop {} {\bibfield  {journal} {\bibinfo  {journal} {Proc.
  R. Soc. Lond. A}\ }\textbf {\bibinfo {volume} {400}},\ \bibinfo {pages} {97}
  (\bibinfo {year} {1985})}\BibitemShut {NoStop}%
\bibitem [{\citenamefont {Watrous}(2018)}]{Watrous2018}%
  \BibitemOpen
  \bibfield  {author} {\bibinfo {author} {\bibfnamefont {J.}~\bibnamefont
  {Watrous}},\ }\href@noop {} {\emph {\bibinfo {title} {{The Theory of Quantum
  Information}}}}\ (\bibinfo  {publisher} {Cambridge University Press},\
  \bibinfo {address} {Cambridge},\ \bibinfo {year} {2018})\BibitemShut
  {NoStop}%
\bibitem [{\citenamefont {Kraus}(1983)}]{Kraus1983}%
  \BibitemOpen
  \bibfield  {author} {\bibinfo {author} {\bibfnamefont {K.}~\bibnamefont
  {Kraus}},\ }\href@noop {} {\emph {\bibinfo {title} {{States, Effects, and
  Operations}}}}\ (\bibinfo  {publisher} {Springer},\ \bibinfo {address}
  {Berlin},\ \bibinfo {year} {1983})\BibitemShut {NoStop}%
\bibitem [{\citenamefont {Choi}(1975)}]{Choi1975}%
  \BibitemOpen
  \bibfield  {author} {\bibinfo {author} {\bibfnamefont {M.-D.}\ \bibnamefont
  {Choi}},\ }\href@noop {} {\bibfield  {journal} {\bibinfo  {journal} {Linear
  Algebra Appl.}\ }\textbf {\bibinfo {volume} {10}},\ \bibinfo {pages} {285}
  (\bibinfo {year} {1975})}\BibitemShut {NoStop}%
\bibitem [{\citenamefont {Birkhoff}(1946)}]{Birkhoff1946}%
  \BibitemOpen
  \bibfield  {author} {\bibinfo {author} {\bibfnamefont {G.}~\bibnamefont
  {Birkhoff}},\ }\href@noop {} {\bibfield  {journal} {\bibinfo  {journal}
  {Univ. Nac. Tucuman Rev. Ser. A}\ }\textbf {\bibinfo {volume} {5}},\ \bibinfo
  {pages} {147} (\bibinfo {year} {1946})}\BibitemShut {NoStop}%
\bibitem [{\citenamefont {Fine}(1982)}]{Fine1982}%
  \BibitemOpen
  \bibfield  {author} {\bibinfo {author} {\bibfnamefont {A.}~\bibnamefont
  {Fine}},\ }\href@noop {} {\bibfield  {journal} {\bibinfo  {journal} {Phys.
  Rev. Lett.}\ }\textbf {\bibinfo {volume} {48}},\ \bibinfo {pages} {291}
  (\bibinfo {year} {1982})}\BibitemShut {NoStop}%
\bibitem [{\citenamefont {Leggett}\ and\ \citenamefont
  {Garg}(1985)}]{LeggettGarg1985}%
  \BibitemOpen
  \bibfield  {author} {\bibinfo {author} {\bibfnamefont {A.~J.}\ \bibnamefont
  {Leggett}}\ and\ \bibinfo {author} {\bibfnamefont {A.}~\bibnamefont {Garg}},\
  }\href@noop {} {\bibfield  {journal} {\bibinfo  {journal} {Phys. Rev. Lett.}\
  }\textbf {\bibinfo {volume} {54}},\ \bibinfo {pages} {857} (\bibinfo {year}
  {1985})}\BibitemShut {NoStop}%
\end{thebibliography}

%

\end{document}